\documentclass{ieeeaccess}

\usepackage{amsmath,amssymb,amsfonts}
\usepackage{algorithmic}
\usepackage{algorithm}
\usepackage{graphicx}
\usepackage{textcomp}

\usepackage{natbib}
\bibpunct[, ]{(}{)}{;}{a}{,}{,}

\usepackage{amsthm}
\usepackage[capitalise]{cleveref}
\usepackage{subfig}
\usepackage{dsfont}

\newcommand{\ket}[1]{\left| #1 \right\rangle}

\newcommand{\ud}[1]{\, \mathrm{d}#1}

\newcommand{\N}{\mathbb{N}}
\newcommand{\Z}{\mathbb{Z}}
\newcommand{\R}{\mathbb{R}}
\newcommand{\C}{\mathbb{C}}
\renewcommand{\P}{ \mathbb P }
\newcommand{\EXP}{\mathbb{E}}
\newcommand{\Var}{\mathbb{V}\mathrm{ar}}

\newcommand{\IND}[1] {{ \mathds{1}_{ #1 }} }

\newcommand{\bs}[1]{\boldsymbol{#1}}

\newcommand{\by}{\bs{y}}

\newcommand{\bK}{\bs{K}}
\newcommand{\bk}{\bs{k}}
\newcommand{\bA}{\bs{A}}
\newcommand{\ba}{\bs{a}}
\newcommand{\bd}{\bs{d}}
\newcommand{\bD}{\bs{D}}

\newcommand{\bDelta}{\bs{\Delta}}
\newcommand{\btheta}{\bs{\theta}}
\newcommand{\bO}{\bs{\Omega}}
\newcommand{\bU}{\bs{\mathcal{U}}}
\newcommand{\modtwo}{\text{mod}\, 2}

\newcommand{\hp}{\hphantom{-}1}

\DeclareMathOperator{\tr}{tr}
\DeclareMathOperator{\diag}{diag}

\DeclareMathOperator{\lat}{lat}
\DeclareMathOperator{\sspan}{span}

\newtheorem{theorem}{Theorem}
\newtheorem{lemma}{Lemma}
\newtheorem{example}{Example}
\newtheorem{corollary}{Corollary}
\newtheorem{definition}{Definition}
\newtheorem{proposition}{Proposition}

\def\BibTeX{{\rm B\kern-.05em{\sc i\kern-.025em b}\kern-.08em
    T\kern-.1667em\lower.7ex\hbox{E}\kern-.125emX}}
\begin{document}
\history{Date of publication xxxx 00, 0000, date of current version xxxx 00, 0000.}
\doi{xxxx}

\title{Expressiveness of Commutative Quantum Circuits:
A Probabilistic Approach}

\author{\uppercase{Jorge M. Ramirez}\authorrefmark{1},
\uppercase{Elaine Wong \authorrefmark{1}, Caio Alves \authorrefmark{1}, Sarah Chehade \authorrefmark{1} and Ryan Bennink}.\authorrefmark{1}}
\address[1]{Oak Ridge National Laboratory, Oak Ridge, TN, 37830, USA}

\tfootnote{This manuscript has been authored by UT-Battelle, LLC under Contract No. DE-AC05-00OR22725 with the U.S. Department of Energy. The United States Government retains and the publisher, by accepting the article for publication, acknowledges that the United States Government retains a non-exclusive, paid-up, irrevocable, world-wide license to publish or reproduce the published form of this manuscript, or allow others to do so, for United States Government purposes. The Department of Energy will provide public access to these results of federally sponsored research in accordance with the Department of Energy's Public Access Plan (http://energy.gov/downloads/doe-public-access-plan). This is the authors' version of the work. It is posted here for your personal use. The definitive version was published in the IEEE Transactions on Quantum Engineering, https://doi.org/10.1109/TQE.2024.3488518.}

\markboth
{Ramirez \headeretal: Expressiveness of Commutative Quantum Circuits}
{Ramirez \headeretal: Expressiveness of Commutative Quantum Circuits}

\corresp{Corresponding author: Jorge M. Ramirez (email: ramirezosojm@ornl.gov).}

\begin{abstract}
This study investigates the frame potential and expressiveness of commutative quantum circuits. Based on the Fourier series representation of these circuits, we express quantum expectation and pairwise fidelity as characteristic functions of random variables, and expressiveness as the recurrence probability of a random walk on a lattice. A central outcome of our work includes formulas to approximate the frame potential and expressiveness for any commutative quantum circuit, underpinned by convergence theorems in probability theory. We identify the lattice volume of the random walk as means to approximate expressiveness based on circuit architecture. In the specific case of commutative circuits involving Pauli-$Z$ rotations, we provide theoretical results relating expressiveness and circuit structure. Our probabilistic representation also provide means for bounding and approximately calculating the frame potential of a circuit through sampling methods.
\end{abstract}


\titlepgskip=-15pt

\maketitle

\section{Introduction}
\label{sec:intro}

\PARstart{Q}{uantum} computing carries the potential to facilitate accelerated computations, enhanced computational efficiency, and the resolution of hitherto intractable problems~\cite{Abhijith2022}.
However, near-term quantum devices are small (10-100 qubits) and prone to noise and other hardware errors, which limits their utility.
This has motivated the development of hybrid algorithms in which a quantum computer and a digital computer work together, each performing the tasks for which they are most suited.
This approach has the potential to facilitate new computational capabilities as qubit number and quality increase~\citep{Preskill2018}.

The vast majority of hybrid algorithms developed to date employ variational quantum circuits~\citep{Cerezo2021}.
In these algorithms, a quantum circuit applies parameterized operations or ``gates'' to a set of qubits to prepare a candidate solution to a problem.
The qubits are then measured and the digital computer uses the measurement results to adjust the circuit parameters in order to improve the state produced by the quantum circuit. 
Variational quantum algorithms have been developed and demonstrated for various problem domains such as optimization, quantum simulation, machine learning~\citep{Grimsley2019, Chowdhury2020, Skolik2022}.

A challenge in variational quantum algorithms is designing circuit forms that can produce a wide range of candidate solutions for a particular problem class~\citep{Du2020}. 
The range of quantum states or unitaries a variational circuit can produce is known as its \emph{expressiveness}~\citep{Du2020}, \emph{expressibility}~\citep{Sim2019, Holmes2022}, or capacity~\citep{Haug2021}. 
The greater the expressiveness, the more likely the circuit is able to produce a good solution to an arbitrary given problem.
Expressiveness is related to computational complexity~\citep{Morales2020} as well as trainability~\citep{Holmes2022}.
Expressiveness has been quantified in various ways, including the rank of the Jacobian matrix~\citep{Haghshenas2022} or quantum Fisher information matrix~\citep{Larocca2023}, distance to a quantum $t$-design~\citep{Sim2019, Holmes2022}, or closeness of a pairwise fidelity distribution to an ideal distribution~\citep[Section 3.1.1]{Sim2019}. In this paper we conceptualize expressiveness in terms of the \emph{frame potential} of the circuit. This approach has been used in \cite{Sim2019}, and is based on the observation that frame potentials can be used to estimate the non-uniformity of the set of states generated by a parametric quantum circuit. 


We compute and analyze the frame potential and expressiveness of a particular family of quantum circuits consisting of mutually commuting operators. These \emph{commutative quantum circuits} are of fundamental interest as they are one of the simplest classes of quantum circuits that cannot be efficiently simulated with digital computers~\citep{Shepherd2009, Bremner2011, Bremner2016}.
They also have shown significant potential as variational ansatze for applications including machine learning~\citep{Havlicek2019}. 

Our approach is probabilistic and leverages the Fourier series representation of commutative quantum circuits, as discussed in~\cite{Vidal2018} and~\cite{Casas2023}. We demonstrate that the quantum expectation and pairwise fidelity of these circuits can be understood as characteristic functions of distinct random variables whose values are related to the spectrum of the associated Hamiltonians. Consequently, we are able to compute expressiveness as the recurrence probability of a random walk on a discrete lattice defined by the circuit form.
This novel probabilistic representation brings to bear an additional set of tools for understanding and computing expressiveness of these circuits.
In contrast to the prevalent numerical methods used to characterize expressiveness in existing literature, our findings are anchored in rigorous mathematical proofs that apply to all instances of commutative quantum circuits. We further undertake a detailed analysis of commutative quantum circuits defined through Pauli-$Z$ rotations. 

Our most significant contribution are formulas to approximate the frame potential and the expressiveness of any commutative quantum circuit. These formulas are rooted in convergence theorems from the realm of probability theory. Our analysis unveils that the relationship between expressiveness and circuit architecture is mostly mediated by the volume of the lattice of the aforementioned random walk. Specializing to commutative quantum circuits whose Hamiltonians are  Pauli-$Z$ rotations, we derive results and provide illustrative examples that show how this volume depends on the circuit, providing insights about how to design circuits of high expressiveness.

We will be considering quantum operators acting on $n$ qubits, and specifically a family of unitaries $U(\btheta)$ parameterized by a vector of circuit parameters $\btheta$ uniformly distributed over $[-\pi,\pi]^N$ for some positive integer $N$, with $n\leq N\leq 2^n$. The initial quantum state is $|\psi_0\rangle$. The following definition introduces the main objects of our study~\citep{Sim2019}.

\begin{definition}\label{def_Exp}
Let $U(\btheta)$ be the unitary operator of an $n$-qubit quantum circuit parameterized by $\btheta \in [-\pi,\pi]^N$ where $n\leq N\leq 2^n$, and has a base quantum state $\psi_0$. The `fidelity' of $U$ is 
\begin{equation}
    F_U(\btheta,\btheta') := \left| \langle \psi_0 | U(\btheta)^\dagger  U(\btheta') | \, \psi_0 \rangle \right|^{2}.
\end{equation}
The `frame potential' of $U$ is
\begin{equation}\label{eq_Ft}
    \mathcal{F}_U(t) := \frac{1}{(2 \pi)^{2N}} \int_{[-\pi,\pi]^{2N}} F_U(\btheta,\btheta')^{t} \ud \btheta \ud \btheta'.
\end{equation}
The `expressiveness' of $U(\btheta)$ is
\begin{equation}\label{eq_defExp}
    \mathcal{E}_U(t) := \frac{\mathcal{F}_\text{Haar}(t)}{\mathcal{F}_U(t)},
\end{equation}
where
\begin{equation}\label{eq_defFHaar}
    \mathcal{F}_\text{Haar}(t) := \binom{2^n+t-1}{2^n-1}^{-1}.
\end{equation}
\end{definition}

The fidelity $F_U \in [0,1]$ in \cref{eq_Ft} quantifies the similarity between states produced by circuits with two different parameter settings. Thus for fixed $t>0$, the frame potential quantifies the average similarity between all pairs of states the circuit can produce.  A relatively small frame potential indicates that most producible states are dissimilar to each other, which intuitively corresponds to large expressiveness.  To obtain a more meaningful measure of expressiveness, we invert the frame potential and normalize it by $\mathcal{F}_\text{Haar}(t)$, the minimum possible frame potential \citep{Sim2019}. This is achieved when the distribution of output states is uniform (Haar invariant) over the entire Hilbert space. With this, $\mathcal{E}_U(t) \in [0,1]$ with larger values corresponding to greater expressiveness.  The order $t$ controls the sensitivity, with greater $t$ corresponding to greater sensitivity. 

In this paper we specialize on commutative quantum circuits. Namely, on unitary operators of the form 
\begin{equation}\label{def_U}
    U(\btheta) = e^{i  \theta_1 H_1} \cdots e^{i \theta_N H_N},
\end{equation} 
where $\btheta=(\theta_1,\dots,\theta_N) \in [-\pi,\pi]^N$ and the Hamiltonians $\{H_j\}_{j=1}^N$ are assumed to be commutative and diagonalizable with integer eigenvalues. In this case, the fidelity can be expressed as 
\begin{equation}\label{eq_FU2Dto1D}
    F_U(\btheta,\btheta') = F_U(\btheta'-\btheta) = \left| \langle \psi_0 | U(\btheta'-\btheta) | \psi_0 \rangle \right|^{2}.
\end{equation}
We thus focus on analyzing the quantum expectation 
\begin{equation}\label{def_fU}
    f_U(\btheta) = \langle \psi_0 | U(\btheta) | \psi_0 \rangle,
\end{equation}
from which the fidelity can be computed simply as $F_U(\btheta) = |f_U(\btheta)|^2$. The function $f_U$ can be written as a Fourier sum
\begin{align}\label{eq:Fourier_sum}
    f_U(\btheta) &= \sum_{\bk \in \bO} \hat{f}_U(\bk) e^{i \btheta \cdot \bk},\\
    \hat{f}_U(\bk) &= \frac{1}{(2\pi)^N}\int_{[-\pi,\pi]^N} f_U(\btheta) e^{-i \btheta . \bk} \ud \btheta
\end{align}
where $\bO$ is the `spectrum' and the `wave vectors' $\bk \in \bO$ are determined by the eigenvalues of the Hamiltonians~\citep{Vidal2018, Schuld2021}. The coefficients $\hat{f}_U(\bk)$ have been characterized for very general circuits in~\cite{Casas2023}, but here we will show that they can be understood as probabilities. Namely, that for commutative circuits, they satisfy $\hat{f}_U(\bk) > 0$ for all $\bk \in \bO$ and $\sum_{\bk \in \bO} \hat{f}_U(\bk) = 1$. Thus, we can write
\begin{equation}\label{eq_fUExp}
    f_U(\btheta) = \EXP e^{i \btheta \cdot K},
\end{equation}
where the random variable $K$ has probability mass function $\P(K = \bk) = \hat{f}_U(\bk)$, for all $\bk \in \bO$. 

The expectation operation $\EXP$ in \cref{eq_fUExp} is not a quantum expectation but rather a mathematical construct that writes $f_U$ as the \textit{characteristic function} of $K$, a random variable taking values on the spectrum $\bO$ of $U$. Our approach is based on this observation and uses tools of probability theory to analyze the fidelity, frame potential, and expressiveness of~$U$. 

The organization of this paper is as follows. In \cref{sec_General} we derive the representation \cref{eq_fUExp} for general commutative circuits. The properties of characteristic functions allow us to obtain in \cref{sec:FPasProb} a representation for the frame potential as $\mathcal{F}_U(t) = \P(W_t=0)$ where $W = \{W_t\}_{t=1}^\infty$ is a random walk on a lattice $\lat(W) \subseteq \Z^N$. In \cref{sec:clt} we apply the central limit theorem to obtain an approximate formula for $\mathcal{F}_U(t)$ in terms of the variance of $K$ and the volume $V_U$ of $\lat(W)$. \cref{sec_CommC,sec:calculations} specialize to commutative quantum circuits where each Hamiltonian $H_j$ is a Pauli-$Z$ rotation on an arbitrary subset of qubits (see \cref{def_HPauli} below). In this case, the circuit architecture is encoded in a binary matrix $\bA$ specifying which qubits are included on each rotation. In \cref{sec:ProbRepCC} we characterize the probabilistic description of $U$ in terms of $\bA$. \cref{sec:characterization,sec:ExtCases} include the derivation of formulas for the volume $V_U$ and some of its properties, along with means for computing bounds for the frame potential of any circuit with $n$ qubits. \cref{sec:calculations} includes illustrative examples and applications of our results to various quantum circuits. It highlights different ways of computing expressiveness, and its behavior as a function of the number of rotations. Lastly, in \cref{sec:outlook} we discuss our results, and provide interesting observations about the possibility of extending our probabilistic framework to non-commutative circuits.

\section{General Commutative Circuits}\label{sec_General}

We consider an operator $U$ of the form \cref{def_U} acting on $n$ qubits, defined by $N\geq n$ Hamiltonians $\{H_j\}_{j=1}^N$, and parameterized by an $N$-dimensional vector $\btheta \in [-\pi,\pi]^N$. Under the assumption that the Hamiltonians in \cref{def_U} are commutative and diagonalizable, it follows that they are simultaneously diagonalizable. Therefore, there exists an operator $Q$ such that $H_j = Q \Lambda_j Q^{\dag}$ for each $j$, where $\Lambda_j$ is a diagonal operator. Moreover, the quantum expectation $f_U$ in \cref{def_fU} can be written as $f_U(\btheta) = \langle  \psi_0 | Q^\dagger e^{i \theta_1 \Lambda_1} \dots e^{i \theta_N \Lambda_N} Q |  \psi_0 \rangle$. We thus assume, without loss of generality, that the operator has the form
\begin{equation}\label{def_Ucomm}
    U(\btheta) = e^{i  (\theta_1 H_1 + \dots + \theta_N H_N)},
\end{equation}
where each Hamiltonian $H_j$ is a diagonal operator with integer eigenvalues,
\begin{equation}
    H_j = \diag(k^{(j)}_1,\dots,k^{(j)}_{2^n}), \quad k_x^{(j)} \in \Z,
\end{equation}
$j=1,\dots, N$, $x = 1,\dots, 2^n.$
We further assume that $n\leq N$ and that the set $\{H_j\}_{j=1}^N$ is linearly independent.

\subsection{Probabilistic Representation}

Establishing that $f_U$ in \cref{def_fU} can be represented as the characteristic function of a random variable is achieved through various methods. First, $f_U$ is a positive-definite function: Let $\{z_k\}_{k\geq 1}$ and $\{\btheta_k\}_{k\geq 1}$ be finite sequences in $\C$ and $[-\pi,\pi]^N$ respectively, then 
\begin{align*}
    \sum_{k,k'} z_k z_{k'}^\dag f_U(\btheta_k - \btheta_{k'}) &= \sum_{k,k'} z_k z_{k'}^\dag f_U(\btheta_k)  f_U(\btheta_{k'})^\dag \\
    &=\left| \sum_{k} z_k f_U(\btheta_k ) \right|^2 \geq 0.
\end{align*}
Bochner's theorem then implies that $f_U$ is a characteristic function \citep{Bochner2005}. Namely, there exists a probability space with a random variable $K$ defined on it, such that \cref{eq_fUExp} holds. We can also directly construct $K$. For this, denote
\begin{equation}\label{def_kx}
    \bk_x := (k_x^{(1)},\dots,k_x^{(N)}), \quad x=1,\dots,2^n,
\end{equation}
and write $(\theta_1 H_1 + \dots + \theta_N H_N) = \diag(\btheta\cdot \bk_1,\dots,\btheta\cdot \bk_{2^n})$. Thus,
\begin{align}
    f_U(\btheta)  &= \tr\left(|\psi_0 \rangle \langle \psi_0| e^{i \diag(\btheta\cdot \bk_1,\dots,\btheta\cdot \bk_{2^n})}\right) \\
    &= \sum_{x=1}^{2^n} |\langle \psi_0|x\rangle|^2 e^{i  \btheta \cdot \bk_x}\\
    &= \EXP \left( e^{i  \btheta \cdot \bk_X}\right), \label{eq_fUEkX}
\end{align}
where $X$ is a random variable taking values on $\{1,\dots,2^n\}$ with $\P(X=x) = |\langle \psi_0|x\rangle|^2$. It will be useful to think of $\bk_X$ as a row chosen randomly from the matrix
\begin{equation}\label{def:bK}
    \bK := \begin{pmatrix}
         -& \bk_1 & -   \\
        & \vdots &\\
         -& \bk_{2^n}& - 
    \end{pmatrix} \in  \Z^{2^n \times N}.
\end{equation}
Lastly, let $\bO$ be the set of rows of $\bK$. In other words, $\bO = \{\bk_x\}_{x=1}^{2^n} \subset \Z^N$, or equivalently the matrix where each row of $\bK$ appears only once, and by our hypothesis having rank equal to $N$. If we define the random variable $K = \bk_X$ and take into account the possible repetitions of each wave vector, we have
\begin{equation}\label{eq_pmfK}
    \P(K = \bk) = \sum_{x: \bk_x = \bk}|\langle \psi_0|x\rangle|^2, \quad \bk \in \bO,
\end{equation}
and the probabilistic representation in \cref{eq_fUExp} holds. 

The moments of $K$ are determined by the Hamiltonians and the base state. The expected value of the $j$-th entry of the vector $K$ is
\begin{equation}
    \EXP K^{(j)} = \sum_{x=1}^{2^n} k_x^{(j)} |\langle \psi_0|x\rangle|^2 = \langle \psi_0 | H_j | \psi_0 \rangle, 
\end{equation}
namely, the quantum expected value of $H_j$ in the state $\psi_0$. Similarly, its covariance matrix is
\begin{equation}
    \Var(K)_{i,j}= \langle \psi_0| H_i H_j |\psi_0 \rangle - \langle \psi_0 | H_j | \psi_0 \rangle \langle \psi_0 | H_i | \psi_0 \rangle,
\end{equation}
for $i,j=1,\dots,N$.

A corresponding probabilistic representation for the fidelity $F_U$ and its powers follows from the properties of characteristic functions. Note that 
\begin{equation}\label{eq_FUExp}
    F_U(\btheta) = f_U(\btheta)f_U(\btheta)^\dag = \EXP e^{i  \btheta \cdot D},
\end{equation}
where $D = K-K'$ and $K,K'$ are assumed to be independent and identically distributed as \cref{eq_pmfK}. We denote the sample space of $D$ as
\begin{equation}\label{def_bDelta}
    \bDelta :=\{ \bk-\bk': \bk,\bk' \in \bO \} \subset \Z^N,
\end{equation}
its cardinality as $L$, and write $\bDelta = \{\bd_l\}_{l=1}^L$. Note that the zero vector $\bs{0}$ is always an element of $\bD$. Also, the distribution of $D$ is symmetric, namely for all $\bd \in \bDelta$, one has $-\bd \in \bDelta$  and $\P(D=\bd) = \P(D=-\bd)$ . It follows in particular that $D$ has mean $\EXP D = 0$ and covariance matrix $\Var(D) = 2\Var(K)$.

For powers of the fidelity, we can write
\begin{equation}
    F_U(\btheta)^t  = \EXP e^{i  \btheta\cdot (D_1+\dots+D_t)}, \quad t=0,1,\dots,
\end{equation}
where $\{D_1,\dots,D_t\}$ are independent and distributed as $D$. We define the random walk
\begin{equation}\label{def_WD}
    W_t = D_1+ \dots +D_t, \quad t=1,2,\dots
\end{equation}
with independent increments distributed as $D$. The paths of $W$ start at the origin and lie on the lattice $\lat(W) \subseteq \Z^N$ generated by $\bDelta$,
\begin{equation}\label{eq_latWspan}
    \lat(W) = \sspan_{\N}(\bDelta) = \sspan_{\Z}(\bDelta), 
\end{equation}
namely the set of linear combinations of elements of $\bDelta$ with coefficients in $\N$ and $\Z$, respectively. The fact that both representations hold for $\lat(W)$ follows from the symmetry of the distribution of $D$. An important quantity associated with $W$ is the volume $V_U$ of its lattice, which can be computed as
\begin{equation}\label{def_VU}
    V_U = |\det(\bDelta^*)|,
\end{equation}
where $\bDelta^*$ is a matrix whose rows are selected from $\bDelta$ and has row space over the integers equal to $\lat(W)$. See \cref{fig:latticepic}.

\begin{figure}[ht]
   \centering
   \includegraphics[scale=0.5]{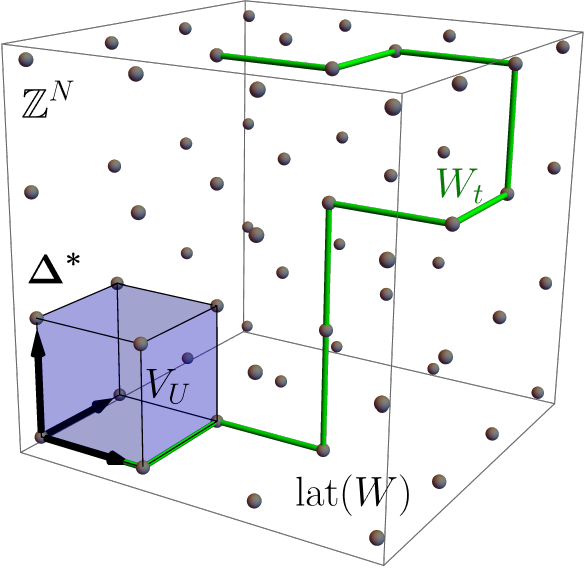}
   \caption{Schematic representation of the lattice $\lat(W) \subseteq \Z^N$ in \eqref{eq_latWspan}, its generator $\bDelta^*$ and volume $V_U$ in \cref{def_VU}, and a path of the random walk $W$ in \eqref{def_WD}.}
   \label{fig:latticepic}
\end{figure}

\subsection{The Frame Potential as a Probability}\label{sec:FPasProb}

Our first result establishes a representation for the frame potential $\mathcal{F}_U$ in terms of the random walk $W$.

\begin{theorem}\label{thm_main}
    Suppose $U$ is given by \cref{def_Ucomm} and let $W$ be the random walk in \cref{def_WD}. The frame potential of $U$ is
    \begin{equation}\label{eq_FPWt0}
        \mathcal{F}_U(t) = \P(W_t = 0), \quad t=1,2,\dots
    \end{equation}
\end{theorem}
\begin{proof}
By virtue of \eqref{eq_FU2Dto1D}, under the change of variables $\bs{\varphi} = \btheta' - \btheta$, $\bs{\varphi'} = \btheta + \btheta'$, the frame potential is  
\begin{align*}
    \mathcal{F}_U(t) =\frac{1}{2^N(2 \pi)^{2N}} \int\int_{R} F_U(\bs{\varphi})^{t} \ud \bs{\varphi} \ud \bs{\varphi}'
\end{align*}
where the region of integration is 
\[
R = \{(\bs{\varphi},\bs{\varphi}') \in [-2\pi,2\pi]^{2N}: |\bs{\varphi}|+|\bs{\varphi}'| \leq 2\pi\}.
\]
The following symmetries of \(F_U\),
\[
    F(\bs{\varphi}) = F(-\bs{\varphi}) = F(\pi+\bs{\varphi}) = F(\pi-\bs{\varphi}),
\]
imply that
\begin{align*}
    \mathcal{F}_U(t) &= \frac{2^{2N}}{(2 \pi)^{2N}} \int\int_{[0,\pi]^{2N}} F_U(\bs{\varphi})^{t} \,d \bs{\varphi} \,d \bs{\varphi}'\\
    &= \frac{1}{\pi^N}  \int_{[0,\pi]^N} F_U(\bs{\varphi})^{t} \,d \bs{\varphi}.
\end{align*}
To compute the resulting integral, note that for \(\by \in \R^N\),
\begin{align*}
    \int_{[0,\pi]^N} & e^{i \btheta \cdot \by} \,  d \btheta \\
    & = \begin{cases}
        \pi^N, & \by=0,\\
        \prod_{j=1}^N \frac{2 \sin(\pi y^{(j)})}{y^{(j)}}, & y^{(j)} \in \R \setminus \Z \text{ for all } j,\\
        0, & y^{(j)} \in \Z \text{ for some } j.
    \end{cases}
\end{align*}
But since \(W_t\) takes only values with integer entries, 
\begin{align*}
    \int_{[0,\pi]^N} F_U(\btheta)^{t} \,d \btheta 
    & = \mathbb{E} \int_{[0,\pi]^N}  e^{i\btheta \cdot W_{t}} \,d \btheta \\
    & = \pi^N \mathbb{P}(W_t = 0)
\end{align*}
from which the result follows.
\end{proof}

As an initial application of Theorem \ref{thm_main}, we prove a useful corollary regarding the monotonicity of the frame potential with respect to the set of rotations. It confirms that augmenting the number of rotations of a circuit decreases the frame potential, which by virtue of \cref{def_Exp}, increases the expressiveness; an intuitive result which has been numerically verified (see Section \cref{sec:calculations}).

\begin{corollary}\label{cor:monotone}
Let the number $n$ of qubits be fixed, and $U$ be the operator in \cref{def_Ucomm}. Consider another circuit $U'$ with $N'>N$ but sharing the first $N$ factors with $U$, namely
\[
    U'(\btheta)=U(\theta_1,\dots,\theta_N)e^{i \theta_{N+1}H_{N+1}} \cdots e^{i  \theta_{N'} H_{N'}},
\]
for some Hamiltonians $H_{N+1},\dots, H_{N'}$ that commute, are diagonalizable, and have integer eigenvalues.  Then $\mathcal{F}_{U'}(t) \leq \mathcal{F}_U(t)$ for all $t$.
\end{corollary}
\begin{proof}
We denote with prime the corresponding quantities associated with the operator $U'$. By construction, the first $N$ components of each $\bk'_x$ agree with those of $\bk_x$ in \cref{eq_fUEkX} for all $x=1,\dots,2^n$. Namely, the matrix $\bK'$ shares the same first $N$ columns with $\bK$ in \cref{def:bK}. Since $K'$ is obtained by randomly selecting rows from $\bK'$, its first $N$ components will have the same distribution as $K$. It follows that the first $N$ components of $W_t'$ have the same distribution as $W_t$, for all $t$. Therefore,
\begin{align*}
    \mathcal{F}_{U'}(t) &= \P(W_t'^{(1)}=0,\dots,W_t'^{(N)}=0,\dots W_t'^{(N')}=0),\\
    &\leq \P(W_t'^{(1)}=0,\dots,W_t'^{(N)}=0),\\
    & = \P(W_t^{(1)}=0,\dots,W_t^{(N)}=0) = \mathcal{F}_U(t).
\end{align*}
\end{proof}

\subsection{Asymptotic Behavior of the Frame Potential and Expressiveness}\label{sec:clt}

The probabilistic interpretation of $F_U$, along with the Central Limit Theorem, imply convergence in distribution of $W_t/\sqrt{t}$ to a multivariate Gaussian vector. Equivalently, by L\'evy's continuity theorem \cite[Section 18.1]{Williams1991}, one has point-wise convergence of the fidelity $F_U$,
\begin{equation}
    \lim_{t \to \infty} F_U(\btheta/\sqrt{t})^{t} = e^{-\frac{1}{2} (\btheta^{\top} \Var(D) \btheta) }, \;
    \btheta \in \R^N,
\end{equation}
where $\Var(D) = 2 \Var(K)$. The asymptotic behavior of the frame potential $\mathcal{F}_
U(t) = \P(W_t = 0)$ follows from an application of the corresponding Local Limit Theorem, for which we follow the theory outlined in \citet[Ch. II.7]{GikhmanSkorohod2004} and \citet[Ch. 2.3]{Lawler2006}.

\begin{theorem}\label{thm_CLTBound}
Suppose $U$ is given by \cref{def_Ucomm}. Let $K$ be the random variable in the probabilistic representation \cref{eq_fUExp} of $f_U$, and $V_U$ the volume of $\lat(W)$. Define the approximate frame potential by
\begin{equation}
    \tilde{\mathcal{F}}_U(t) = \frac{V_U}{(4\pi t)^{N/2} \sqrt{\det(\Var(K))}}.
\end{equation}
Then there exists $c>0$ such that $|\mathcal{F}_U(t) - \tilde{\mathcal{F}}_U(t)| < c ~t^{-(N/2+1)}$ for all $t>0$.
\end{theorem}

\begin{proof}
If the Hamiltonians form a linearly independent set of operators, then both $\bK$ and $\bDelta$ have rank $N$ and the matrix $\bDelta^* \in \Z^{N\times N}$ in \cref{def_VU} is invertible. Denote $\bs{\Psi} = (\bDelta^*)^{-1}$ and consider the transformed random walk $\tilde{W}_t := W_t \bs{\Psi}$ with increments in the set $\tilde{\bDelta}:= \{\bd \bs{\Psi}, \bd \in \bDelta\} =: \{\tilde{\bd}_l\}_{l=1}^L$. We will show that the transformed random walk $\tilde{W}$ is `completely irreducible'. First note that by construction, each $\tilde{\bd}_l$ contains the coefficients to express $\bd_l$ as a linear combination of elements of $\bDelta$, and therefore $\tilde{\bd}_l \in \Z^N$. Also, $\tilde{\bDelta}$ is symmetric and contains the canonical basis for $\Z_N$. Therefore $\lat(\tilde{W}) = \Z^N$ and $\tilde{W}$ is an `irreducible' random walk. To prove complete irreducibility, we must prove reducibility for a random walk with increments on $\tilde{\bDelta} - \tilde{\bd}$ where $\tilde{\bd}$ is an arbitrary vector of $\tilde{\bDelta}$. This reduces to showing that $\sspan_{\Z}(\tilde{\bDelta} - \tilde{\bd}) = \Z^N$. Clearly, $\sspan_{\Z}(\tilde{\bDelta} - \tilde{\bd}) \subset \sspan_{\Z}(\tilde{\bDelta}) = \Z^N$. Let $\bk \in \sspan_{\Z}(\tilde{\bDelta})$, then we can write for some $m_1,\dots,m_L \in \Z$,
\[
    \bk = \sum_{l=1}^L m_l \tilde{\bd_l} = \sum_{l=1}^L m_l (\tilde{\bd}_l - \tilde{\bd}) - m_l (\bs{0} - \tilde{\bd})  
\]
which is a linear combination of vectors in $(\tilde{\bDelta} - \tilde{\bd})$ because $\bs{0} \in \bDelta$. Then $\bk \in \sspan_{\Z}(\tilde{\bDelta} - \tilde{\bd})$. Since $\tilde{W}$ is completely irreducible, the Local Limit Theorem 7.2 in \cite{GikhmanSkorohod2004} and Theorem 2.3.5 of \cite{Lawler2006} apply. Specifically, there exists a constant $0<c<\infty$ such that
\[
     \left| \P(\tilde{W_t} = 0) -\frac{1}{(2 \pi t)^{N/2} \sqrt{\det(\Var(\tilde{D}))}}\right| \leq \frac{c}{t^{\frac{N}{2}+1}}.
\]
Note that $\P(\tilde{W_t} = 0) = \P(W_t = 0)$. Also, from \cref{eq_momsK} and the definition of $\bs{\Psi}$ follows that
\begin{align*}
    \det(\Var(\tilde{D})) = \det( \bs{\Psi}^\dag \Var(D) \bs{\Psi}) = \frac{2^N \det(\Var(K))}{V_U^2}, 
\end{align*}
which completes the proof.
\end{proof}

The following asymptotic result for the expressiveness follows from combining \cref{thm_CLTBound} with the asymptotic behavior of $\mathcal{F}_{\text{Haar}}(t)$ in \cref{def_Exp}.
\begin{corollary}\label{cor_exp}
    Define the approximate expressiveness by
    \begin{equation}
        \tilde{\mathcal{E}}_U(t) = \frac{\sqrt{\det(\Var(K))}}{V_U} (4 \pi t)^{N/2} \mathcal{F}_{\text{Haar}}(t).
    \end{equation}
    Then, under the assumptions of Theorem \ref{thm_CLTBound} there exists a positive constant $c$ such that $|\tilde{\mathcal{E}}_U(t) - \mathcal{E}_U(t)| < c \,t^{\frac{N}{2} - 2^n}$ for all $t>0$.
\end{corollary}

\section{Commutative Circuits with Pauli Rotations }
\label{sec_CommC}

In order to obtain more detailed insights into the relationship between operators, frame potentials, and expressiveness, we consider a specific type of Hamiltonian. We fix the initial state $\ket{\psi_0} = |+ \rangle^{\otimes n}$ and consider an operator $U(\btheta)$ composed exclusively of Pauli-$Z$ rotations on arbitrary sets of qubits. Namely, as in \cref{def_Ucomm} with the $j$-th Hamiltonian $H_j$ acting on a non-empty set of qubits $Q_j \subseteq \{1,\dots,n\}$ and given by 
\begin{equation}\label{def_HPauli}
    H_{j} = \bigotimes_{m=1}^n O^{(j)}_{m},
\end{equation}
where
\begin{equation}\label{def_Ojm}
O^{(j)}_{m} = \begin{cases}
    Z, & m \in Q_{j} \\
    I_2, & m \notin Q_{j}
\end{cases}\end{equation}
and
\begin{equation}
Z = \begin{pmatrix} 1 & \hphantom{-}0\\ 0 & -1 \end{pmatrix}, \quad
I_2 = \begin{pmatrix} 1 & 0\\ 0 & 1 \end{pmatrix}.    
\end{equation}
See \cref{fig:ccircuit} for an example. We further assume that all Hamiltonians are distinct, and that
$n \leq N \leq 2^n-1$, with the case $N=2^n-1$ corresponding to a circuit with all possible rotations. In this setup, the amplitude of the base state is $\langle \psi_0|x\rangle = 2^{-n/2}$, and the eigenvalues of the Hamiltonians satisfy $k_x^{(j)} \in \{-1,1\}$ for all $x=1,\dots,2^n$, $j=1,\dots,N$. 

Further, the resulting operator $U(\btheta)$ can be encoded in a binary matrix $\bA \in \{0,1\}^{n \times N}$ having entries $a_m^{(j)} = 1$ whenever $m$-th qubit is included in the $j$-th rotation, namely
\begin{equation}\label{def_A}
    a_m^{(j)} := \begin{cases}
    1, & m \in Q_j \\
    0, & m \notin Q_j
\end{cases}, \quad m=1,\dots,n, \; j =1,\dots,N.
\end{equation}

\Figure[t!](topskip=0pt, botskip=0pt, midskip=0pt){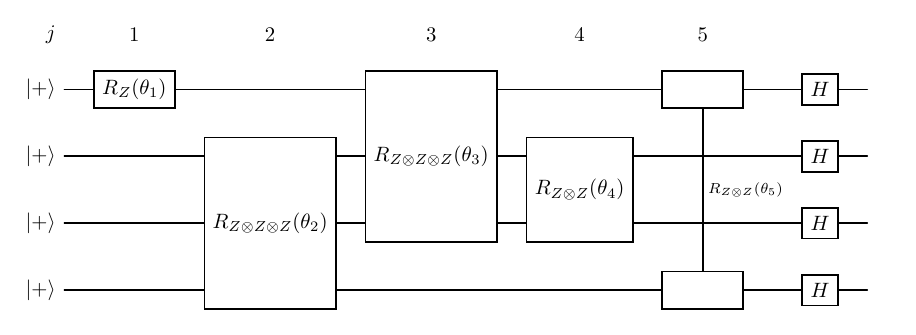}
{An example of a commutative circuit with $n=4$ and $N=5$.\label{fig:ccircuit}}

\subsection{Probabilistic Representation}\label{sec:ProbRepCC}

We now apply the probabilistic representation developed in \cref{sec_General} to circuits with Hamiltonians as in \cref{def_HPauli}, and characterize some of their properties in terms of the matrix $\bA$. Our first result relates algebraically the matrices $\bK$~and~$\bA$.

\begin{proposition}\label{prop_KUA}
Suppose $\{H_j\}_{j=1}^N$ are given by \cref{def_HPauli} and let $\bK$ be as in \cref{def:bK}. Then each column of $\bK$ is a column of the Hadamard matrix $\bs{H}_{2^n}$ obtained by the  Sylvester construction. Moreover, $\bK$ has rank $N$ and it is given by
\begin{equation}\label{eq_2KUA}
    \bK=\bs{1}-2 \left(\bU\bA\right)_{\modtwo} \in \{-1,1\}^{2^n \times N}, 
\end{equation}
where $\bU$ is a $2^n\times n$ binary matrix whose rows are all possible $n$-tuples of $\{0,1\}$.
\end{proposition}

\begin{proof}
Each column $\bk^{(j)}$ of $\bK$ can be explicitly written as a Kronecker product of vectors corresponding to the diagonals $o_m^{(j)}$ of $O_m^{(j)}$ in \cref{def_Ojm} as
\begin{equation}\label{eq_Kjxi}
    \bk^{(j)} =  \bigotimes_{m=1}^n o^{(j)}_{m},
\end{equation}
where for each $m=1,\dots,n$,
\begin{equation}\label{def_xijm}
o^{(j)}_{m}=\begin{cases}
\diag{(Z)}, & m \in Q_{j} \\
\diag{(I_2)}, & m \notin Q_{j}
\end{cases}.
\end{equation}
It follows from \cref{def_xijm} that each column of the matrix $\bK$ is a column from the Hadamard matrix $\bs{H}_{2^n}$ of order $n$ obtained by Sylvester's construction. See \cite{Horadam2012} for reference. In particular, the $N$ columns of $\bK$ form an orthogonal set of $\Z^{2^n}$ and are therefore linearly independent. To establish \cref{eq_2KUA}, refer to the construction in \cref{def_Ojm} and observe that the correspondences
\[
\diag{(Z)}=
\begin{pmatrix}\hp\\-1\end{pmatrix}
=
\begin{pmatrix}(-1)^0\\(-1)^1\end{pmatrix} 
\]
and
\[
\diag{(I_2)}=
\begin{pmatrix}1\\1\end{pmatrix}
=
\begin{pmatrix}(-1)^0\\(-1)^0\end{pmatrix}
\]
maps the vectors $o_m^{(j)}$ in \cref{def_xijm} to
\[
(-1)^{\tiny \begin{pmatrix} 0 \\ a_m^{(j)}\end{pmatrix}}, \quad m=1,\dots,n, \quad j=1,\dots,N,
\]
with $a_m^{(j)}$ denoting the entries of $\bA$ as in \cref{def_A}. Since products of powers of $(-1)$ correspond to addition modulus 2 of the exponents, expanding the tensor product in \cref{eq_Kjxi} gives $(-1)$ raised to a linear combination of the columns of $\bU$ with coefficients $a_m^{(j)}$. In other words, for the $j$-th column,
\begin{equation*}
    \bk^{(j)} = -1^{\left(\bU \ba^{(j)}\right)_{\modtwo}} = \bs{1}-2 (\bU \ba^{(j)})_{\modtwo}, \quad j=1,\dots,N,
\end{equation*}
and this completes the proof.
\end{proof}

The random index $X$ in \cref{def_kx} can take $2^n$ values, each with probability $\langle \psi_0 | x\rangle^2 = 1/2^n$. In order to describe the random variable $K$ and its state space $\bO$ in \cref{thm_main}, we need to characterize the possible repetitions of each $\bk$ in $\bK$. 

\begin{theorem}\label{thm_distK}
Suppose $\{H_j\}_{j=1}^N$ are given by \cref{def_HPauli}. The random vector $K$ is uniformly distributed on 
\begin{equation}\label{eq_Omrow}
    \bO  = \bs{1}-2\mathrm{row}_{\Z_2^N}(\bA),
\end{equation}
where $\mathrm{row}_{\Z_2^N}(\bA)$ denotes the row space of $\bA$ in $\Z_2^N$. The distribution of $K$ is
\begin{equation}\label{eq_distKj}
    \P\left(K^{(j)} = 1\right) = \P\left(K^{(j)} = -1\right) = \frac{1}{2},
\end{equation}
for $j=1,\dots,N$, and $K$ has expectation and covariance matrices respectively given by
\begin{equation}\label{eq_momsK}
\EXP K = 0, \quad \Var K = I_N.
\end{equation}
\end{theorem}
\begin{proof}
By \cref{eq_2KUA} in Lemma \ref{prop_KUA}, each row $\bk_x$ of $\bK$ corresponds to the linear combination $(\bU_x \bA)_{\modtwo}$ of the rows of $\bA$ as elements of $\Z_2^N$. Since $\bU$ exhausts all binary $n$-tuples, it follows that $(\bU \bA)_{\modtwo}$ contains all vectors in $\mathrm{row}_{\Z_2^N}(\bA)$, from which \cref{eq_Omrow} follows. The cardinality of $\bO$ is $\#(\bO)  = 2^{r}$ with $r:= \mathrm{rank}_{\Z_2^N}(\bA)$ and $r\leq n \leq N$. Each row $\bk_x$ of $\bK$ corresponds to the solutions $\bU_x$ of
\[
(\bU_x \bA)_{\modtwo} = \frac{1}{2}(\bs{1}-\bk_x).
\]
The number of such solutions equals two raised to the nullity of $\bA$, namely $2^{n-r}$. Hence each element of $\bO$ is repeated $2^{n-r}$ times as a row in $\bK$. It follows that for all $\bk$,
\begin{equation*}
    \P(K=\bk) = \frac{\#\{x:\bk_x=\bk\}}{2^n}= \frac{2^{n-r}}{2^{n}} = 2^{-r}
\end{equation*}
and $K$ is therefore uniformly distributed on $\bO$. Further, note that by \cref{prop_KUA}, each column of $\bK$ is a column of the Hadamard matrix $\bs{H}_{2^n}$, and therefore contains exactly $2^{n-1}$ entries equal to $1$ and $-1$ respectively from which it follows that \cref{eq_distKj} holds and $\EXP(K) = 0$. Also,
\begin{equation*}
    \Var(K^{(j)}) = \sum_{x=1}^{2^n} (k^{(j)}_x)^2 \langle +|x\rangle^2 = 1.
\end{equation*}
Moreover, the columns of $\bs{H}_{2^n}$ form an orthogonal set of vectors in $\{-1,1\}^{2^n}$, and hence so do the columns of $\bK$. This proves that $\Var(K)_{j_1,j_2} = 0$ for all $j_1 \neq j_2$.
\end{proof}

\begin{example}
\label{ex_omega}
For the circuit in Fig. \ref{fig:ccircuit}, we have $n=4,N=5$,
\begin{small}
\[
\bA\!=\!\begin{pmatrix}
1 & 0 & 1 & 0 & 1 \\
0 & 1 & 1 & 1 & 0 \\
0 & 1 & 1 & 1 & 0 \\
0 & 1 & 0 & 0 & 1
\end{pmatrix}\!\!,\,
\bK\!=\!
\begin{pmatrix}
\hp & \hp & \hp & \hp & \hp \\
\hp & -1 & \hp & \hp & -1 \\
\hp & -1 & -1 & -1 & \hp \\
\hp & \hp & -1 & -1 & -1 \\
\hp & -1 & -1 & -1 & \hp \\
\hp & \hp & -1 & -1 & -1 \\
\hp & \hp & \hp & \hp & \hp \\
\hp & -1 & \hp & \hp & -1 \\
-1 & \hp & -1 & \hp & -1 \\
-1 & -1 & -1 & \hp & \hp \\
-1 & -1 & \hp & -1 & -1 \\
-1 & \hp & \hp & -1 & \hp \\
-1 & -1 & \hp & -1 & -1 \\
-1 & \hp & \hp & -1 & \hp \\
-1 & \hp & -1 & \hp & -1 \\
-1 & -1 & -1 & \hp & \hp
\end{pmatrix}. 
\]
\end{small}
Note that $\bK$ has $2^n = 16$ rows and its columns are orthogonal. The rank of $\bA$ is $r=\mathrm{rank}_{\Z_2^n}(\bA)=3$ and $\bK$ has $2^r=8$ distinct rows and each repeated $2^{n-r}=2$ times.
\end{example}

Next, we turn to the difference $D$ and the random walk $W$ in \cref{def_WD}. First note that $D=K-K' \in \{-2,0,2\}^N$, consequently $\bDelta \in \{-2,0,2\}^{L \times N}$. The lattice $\lat(W)$ in \cref{eq_latWA} is thus a subset of $2\Z^N$. The following Proposition characterizes this lattice in terms of $\bA$.

\begin{proposition}\label{prop_relate2A}
Suppose $\{H_j\}_{j=1}^N$ are given by \cref{def_HPauli}.  The lattice of $W$ can be written 
\begin{equation}\label{eq_latWA}
    \lat(W) = \sspan_{\Z}(\bO - \bs{1}) = \sspan_{\Z}(2\mathrm{row}_{\Z_2^N}(\bA)),
\end{equation}
where $\bO - \bs{1}$ denotes the set (or matrix) obtained from subtracting 1 from every entry of the vectors in $\bO$.
\end{proposition}
\begin{proof}
The rightmost equation in \cref{eq_latWA} follows directly from \cref{eq_Omrow} in \cref{thm_distK}. For the equation on the left, note that the constant vector $(1,\dots,1) \in \Z^N$ is always an element of $\bO$ because the first row of the Hadamard matrix $\bs{H}_{2^n}$ is composed of only ones. Since $\bDelta$ are differences between elements of $\bO$, it follows that $\sspan_{\Z}(\bO-\bs{1}) \subseteq \sspan_{\Z}(\bDelta) = \lat(W)$. Any element in $\lat(W)$ can be written as
\begin{equation*}
    \sum_{l=1}^L m_l (\bk_l - \bk'_l) = \sum_{l=1}^L m_l (\bk_l - \bs{1}) - m_l (\bk'_l - \bs{1}),
\end{equation*}
which is a linear combination of elements of $\bO - \bs{1}$, from which \cref{eq_latWA} follows.  
\end{proof}

\subsection{Lattice Volume and the Frame Potential}
\label{sec:characterization}

A direct application of \cref{thm_CLTBound} and \cref{cor_exp} gives the following expressions that approximate the frame potential and expressiveness of $U$ in terms of its corresponding lattice volume.

\begin{proposition}\label{prop:FEapproxComm}
    Suppose $\{H_j\}_{j=1}^N$ are given by \cref{def_HPauli}, then the approximate frame potential and expressiveness of $U$ are, respectively, 
\begin{equation}\label{eq:FEtildeComm}
    \tilde{\mathcal{F}}_U(t) = \frac{V_U}{(4\pi t )^{N/2}}, \quad 
    \tilde{\mathcal{E}}_U(t) = 
        \frac{(4 \pi t)^{N/2} \mathcal{F}_{\text{Haar}}(t)}{V_U}. 
\end{equation}
\end{proposition}

\begin{example}
    For the circuit shown in \cref{fig:ccircuit} with $n=4$, $N=5$, the value of the lattice volume is $V_U = 128$. \cref{fig:FEexample} compares the exact frame potential and expressiveness with their corresponding approximate formulas in \cref{eq:FEtildeComm}. 
\end{example}

In fact, the lattice volume $V_U$ completely determines the behavior of the frame potential for large values of $t$. We aim at characterizing $V_U$ as a function of the circuit's architecture. The ultimate goal is to understand which operators give rise to small frame potentials (as a function of $t$) and therefore to large expressiveness for a given amount of resources. First we prove a lemma with two distinct representations of the lattice volume. 

 \begin{lemma}\label{lem_minor}
    Let $U$ be an operator on $n$ qubits with $\{H_j\}_{j=1}^N$ given by \cref{def_HPauli} with $n \leq N \leq 2^n-1$. The lattice volume is given by
    \begin{equation}\label{eq_vol2NUA}
    V_U = 2^N |\det((\bU \bA)_{\modtwo}^*)|,
    \end{equation}
    where $(\bU \bA)_{\modtwo}^*$ is an $N\times N$ matrix whose rows are selected from $(\bU \bA)_{\modtwo}$ and span $\lat(W)$. Also, there exists an $N\times N$ minor, denoted $\mathrm{Min}^*(\bs{H}_{2^n})$, of the Hadamard matrix $\bs{H}_{2^n}$ such that
    \begin{equation}\label{eq_thmMinor}
        V_U = |\mathrm{Min}^*(\bs{H}_{2^n})|.
    \end{equation}
\end{lemma}
\begin{proof}
    \cref{eq_vol2NUA} follows from \cref{eq_latWA} in \cref{prop_relate2A} because $\mathrm{row}_{\modtwo}(\bA)$ is the set of rows of $(\bU \bA)_{\modtwo}$. For \cref{eq_thmMinor}, recall that $\bO$ is the matrix obtained by selecting $N$ columns from of $\bs{H}_{2^n}$ and deleting the repeating rows. Consider the matrix $\bO' := (\bO -\bs{1})^* + \bs{1}$, where the operation $*$ selects $N$ rows of $(\bO -\bs{1})$ generating the lattice. Then $\bO'$ is also a sub-matrix of $\bs{H}_{2^n}$. Also, \cref{eq_latWA} implies $V_U = |\det(\bO' - \bs{1})|$. Consider the \mbox{$(N+1)\times (N+1)$} matrix $\bO''$ obtained by appending a row vector of $1$s to the top of $\bO'$ and a column vector of $1$s to the left of $\bO'$. Then $\bO''$ is also a sub-matrix of $\bs{H}_{2^n}$ because the first row and column of $\bs{H}_{2^n}$ have all elements equal to one. If for all $i=2,\dots,N+1$ we replace the $i$-th row of $\bO''$ by itself minus the constant vector of $1$s, the resulting matrix $\bO'''$ has the same determinant as $\bO''$ with entries $\bO'''_{1,1} = 1$, $\bO'''_{j,1} = 0$ and $\bO'''_{i,j} = (\bO'-\bs{1})_{i-1,j-1}$ for $i,j=2,\dots,N+1$. Taking cofactors with respect to the first column in $\bO'''$, we obtain
    \begin{equation*}
        \det(\bO'') = \det(\bO''') = \det(\bO'-\bs{1}) = \det((\bO-\bs{1})^*).
    \end{equation*}
    But $\det(\bO'')$ is a minor of $\bs{H}_{2^n}$ and therefore $V_U = |\det(\bO'')|$ satisfies \cref{eq_thmMinor} as desired.
\end{proof}

The minors of any Hadamard matrix satisfy the famous `Hadamard bound' which states that any $N\times N$ matrix~$\bs{M}$ with entries in the complex unit disk, satisfies $|\det(\bs{M})| \leq N^{N/2}$. Hadamard matrices are precisely those for which the bound is realized, in particular, for Sylvester's construction, $|\det(\bs{H}_{2^n})| = (2^n)^{2^{n-1}}$. See \cite{Kravv2016} and references therein. For example, it follows directly from \cref{eq_vol2NUA} and the Hadamard bound, that $V_U \leq 2^N N^{N/2}$ for all~$N$. The following result refines this bound and provides a more comprehensive picture of the volume. 

\begin{theorem}\label{thm:maxminvol}
Let $U$ be an operator on $n$ qubits with $\{H_j\}_{j=1}^N$ given by \cref{def_HPauli}  with $n \leq N \leq 2^n-1$. Then the lattice volume, $V_U$, satisfies the following bound
\begin{align}\label{eqn_boundvU}
    V_{\min} \leq 2^N \leq V_U \leq (N+1)^{\frac{N+1}{2}} \leq   V_{\max},
\end{align}
where $V_{\min} := 2^n$ and $V_{\max} := (2^n)^{2^{n-1}}$. The maximum possible value of the volume $V_{\max}$ is attained by a circuit with all rotations, $N= 2^n-1$.
The minimum possible value of the volume $V_{\min}$ is attained by a circuit with $N=n$ which occurs if and only if $\bA$ is a square matrix of full rank in $\Z_2^N$. 
\end{theorem}

\begin{proof}
The left most inequality is trivial since by hypothesis $n\leq N$, which implies $2^n\leq 2^N$. For the next inequality, since $|\det((\bU  \bA)^*_{\modtwo})|>1$, it follows that $2^N\leq 2^N|\det((\bU  \bA)^*_{\modtwo})| = V_U$, by \cref{eq_vol2NUA}. The third inequality follows directly from the Hadamard bound applied to the $V_U = |\det(\bO'')|$ where $\bO''$ is the the $(N+1)\times (N+1)$ matrix constructed in \cref{lem_minor}. The last inequality follows from monotonicity of the function $N \mapsto N^{N/2}$. 

For the case of minimum possible volume, recall that $\bA \in \{0,1\}^{n \times N}$ with $N \geq n$. If $N=\mathrm{rank}_{\Z_2^N}(\bA)$, then $N=n$ and $\mathrm{row}_{\Z_2^N}(\bA) =\Z_2^N$.  \cref{eq_Omrow} implies that $\bO-\bs{1}$ contains all elements of $\{-2,0\}^N$ and therefore by \cref{eq_latWA}, $V_U = 2^N$. 
Conversely, if $N>n$ the left null space of  $\bA$ in $\Z_2^N$ is not trivial, and $\mathrm{row}_{\Z_2^N}(\bA) \subset \Z_2^N$. But the set of rows of $(\bU  \bA)_{\modtwo}$ is exactly $\mathrm{row}_{\Z_2^N}(\bA)$. Therefore, by selecting rows that generate $\lat(W)$, one obtains a matrix $(\bU  \bA)^*_{\modtwo}$ that is row deficient in $\Z_2^N$, namely $
\det((\bU  \bA)^*_{\modtwo}) = 0 ~\modtwo$. But since $(\bU  \bA)^*_{\modtwo}$ is invertible in $\Z^N$, then $\det((\bU  \bA)^*_{\modtwo})$ is a multiple of two, and $V_U > 2^N$. For the case of maximum volume, note that by Hadamard's theorem $|\det(\bO'')| = V_{\max}$ only if $\bO''$ is a $2^n \times 2^n$ Hadamard matrix. The construction in the proof of \cref{lem_minor} thus implies that the number of rotations is the maximum possible $N = 2^n-1$.
\end{proof}

Note that the smallest value of the volume $V_U$ does not correspond necessarily with the smallest possible frame potential because of the presence of the factor $t^{-N/2}$ in \eqref{eq:FEtildeComm}. 
In \cref{sec:calculations} we will explore the numerical behavior of the frame potential and the expressiveness as a function of $t,n$ and $N$, and its relationship with the lattice volume. 

\subsection{Extreme Cases}\label{sec:ExtCases}

For the two extreme cases $N=n$ and $N=2^n-1$ described in  \cref{thm:maxminvol}, our probabilistic representation yields means of computing the frame potential. We first consider the case of all rotations, which as noted in~\cite{Bennink2023} can be approached as the problem of counting the number of abelian squares of size $t$ with an alphabet of size $2^n$. We include the proof here for completeness and because our approach is different from the one taken in~\cite{Bennink2023}.

\begin{theorem}\label{thm:FUExpMultinomial}
If the circuit includes all possible rotations, $N=2^n-1$, the frame potential can be computed as the expected value of a multinomial coefficient, 
\begin{equation}\label{eqn:FUExpMultinomial}
    \mathcal{F}_U(t) = \frac{1}{2^{nt}} \EXP \left[\binom{t}{B^{(1)}_t \, \cdots B^{(2^n)}_t}\right],
\end{equation}
where $B_t = (B_t^{(1)},\dots,B_t^{(2^n)})$ is a random walk with increments uniformly distributed on the canonical basis for $\R^{2^n}$.     
\end{theorem}
\begin{proof}
Let $S_t$ and $S'_t$ be independent random variables distributed as $K_1+\dots+ K_t$. Then we can write $\mathcal{F}_U(t) = \P(S_t = S_t')$. Further, $S_t = B_t \bK$, where $B_t$ is a row vector of $2^n$ non-negative integers counting how many times in $S_t$ has each $\bk \in \bK$ appeared,
\[
    B_t^{(x)} = \# \{s\leq t: K_s = \bk_x\}, \quad x=1,\dots,2^n.
\]
At each time-step $t$, the row vector $B_t$ increases by one at a single entry uniformly selected from $2^n$ choices, which is equivalent to a random walk with increments on the canonical basis for $\R^{2^n}$. Also, each $B_t$ is distributed as a multinomial with $t$ trials and $2^n$ equally likely outcomes, 
\[
    \P(B_t = (b_1,\dots,b_{2^n})) = \frac{1}{2^{nt}}\binom{t}{b_1 \cdots b_{2^n}},
\]
with $b_1,\dots,b_{2^n} \in \{0,1,2,\dots\}$. Defining $B'_t$ analogously for $S_t'$, we can write,
\[
    \mathcal{F}_U(t) = \P((B_t -B'_t) \in \text{null}(\bK^\top)).
\]
Recall from \cref{prop_KUA} that $\bK$ has rank equal to $N$, so the  dimension of $\text{null}(\bK^\top)$ equals $2^n-N$. In the special case $N = 2^n-1$, $\bK$ exhausts every column of $\bs{H}_{2^n}$ except for the column containing only ones (which would correspond to a rotation gate where the identity operator is applied to each qubit). As a consequence,   $\text{null}(\bK^\top)$ has dimension one and contains only constant vectors. Further, $\sum_{x=1}^{2^n} B_t^{(x)} = \sum_{x=1}^{2^n} {B'}_t^{(x)} = t$, so $(B_t -B'_t) \in \text{null}(\bK^\top)$ is equivalent to $B_t = B'_t$, and
\begin{align}
    \mathcal{F}_U(t) &= \P(B_t = B_t') = \sum_{\bs{b}} \P(B_t = \bs{b})^2\\ 
    &= \frac{1}{2^{nt}} \EXP \left[\binom{t}{B^{(1)}_t \, \cdots B^{(2^n)}_t}\right]
\end{align}
where the sum is taken over all possible values $\bs{b}$ of $B_t$.
\end{proof}

In the case of the minimum possible number of rotations, the frame potential can be explicitly computed. Note that by the monotonicity property established in \cref{cor:monotone}, the frame potential in this case is an upper bound for the frame potential of any circuit with the same number of qubits.

\begin{theorem}
If the circuit has the minimum number of rotations $N=n$, then the frame potential is explicitly given for all $t$ by
\begin{equation*}
    \mathcal{F}(t) = \left(\frac{\Gamma(t+\tfrac{1}{2})}{\sqrt{\pi} \Gamma(t+1)}\right)^n
\end{equation*}
where $\Gamma(\cdot)$ denotes the Gamma function.
\end{theorem}

\begin{proof} Suppose $r=n=N$. Then $K = (K^{(1)},\dots,K^{(N)})$ is a vector of independent and identically distributed random variables. To see that, note that if  $r=N$, by \cref{eq_distKj} in the proof of Theorem \ref{prop_KUA},
\begin{equation*}
    2^{-N} = \P(K^{(1)} = k^{(1)}) \cdots \P(K^{(N)} = k^{(N)}) = \P(K = \bk)
\end{equation*}
for any $\bk=(k^{(1)},\dots,k^{(N)}) \in \bK$. Likewise, the vectors $D_t$ and $W_t$ in \cref{def_WD} have independent and identically distributed components. In this case we can compute
\[
    \mathcal{F}(t) = \P(W_t = 0) = \prod_{j=1}^N \P(W_t^{(j)} = 0) = \P(W_t^{(1)}=0)^N ,
\]
where $W_t^{(1)} = D_1^{(1)} + \dots + D_t^{(1)}$ is a random walk with increments distributed as
\begin{equation*}
    \P(D_1^{(1)} = \pm1) = \frac{1}{4}, \quad \P(D_1^{(1)} = 0) = \frac{1}{2}.
\end{equation*}
We can compute $\P(W_t^{(1)} = 0)$ by counting how many elements of $\bD_t^{(1)} = \{D_1^{(1)}, \dots, D_t^{(1)}\}$ fall on each of $\{-1,0,1\}$, namely, $C_{-1}(\bD_t^{(1)})$, $C_{0}(\bD_t^{(1)})$, $C_{1}(\bD_t^{(1)})$. Consider the case when $t$ is even. Then the event $[W_t^{(1)} = 0]$ implies that $C_{0}(\bD_t^{(1)})$ is even, say $C_{0}(\bD_t^{(1)}) = 2c_0$, and that $C_{-1}(\bD_t^{(1)}) =  C_{1}(\bD_t^{(1)}) = (t-2c_0)/2$. Thus we can write the probability in terms of multinomial coefficients as follows:
\begin{align*}
    \P(W^{(1)}_t = 0) &= \sum_{c_0 = 0}^{t/2} \frac{t!}{(2c_0)! (\tfrac{t-2c_0}{2})! (\tfrac{t-2c_0}{2})!} \frac{1}{2^{2c_0}} \frac{1}{4^{t-2c_0}}\\
    &= \frac{\Gamma(t+\tfrac{1}{2})}{\sqrt{\pi} \Gamma(t+1)}
\end{align*}
which yields the exact formula for $\mathcal{F}(t)$. 
\end{proof}

\section{Illustrative Calculations}
\label{sec:calculations}

The probabilistic framework developed in \cref{sec_General} and \cref{sec_CommC} can be used to design methods for computing the frame potential and the expressiveness with different levels of accuracy. This is significant because  $\mathcal{F}_U(t)$ becomes a very small quantity as $t$ increases, making its accurate numerical computation very difficult. In this section, we present examples of such methods on circuits with Pauli-$Z$ rotations to illustrate some applications of our results, and highlight the usefulness of the numerical approximation obtained in \cref{thm_CLTBound}.

As in \cref{sec_CommC}, let $U(\btheta)$ be a parameterized commutative circuit in the form of \cref{def_Ucomm} with Hamiltonian operators given by Pauli-$Z$ rotations defined by a matrix $\bA$. Then the frame potential of $U$ is, by \cref{thm_main}, $\mathcal{F}_U(t) = \P(W_t = 0)$ with $W$ as in \cref{def_WD}. To begin with, we can use \cref{thm_distK} to compute
\begin{align}
    \mathcal{F}_U(1) &= \P(W_1 = 0) = 2^{-\mathrm{rank}_{\Z_2^N}(\bA)}\\
    \mathcal{E}_U(1) &= 2^{\mathrm{rank}_{\Z_2^N}(\bA)-n}
\end{align}
which directly relate the frame potential and expressiveness of $U$ at $t=1$ to the composition of the circuit.

For sufficiently small values of $n,N$ and $t$, the frame potential can be computed exactly by performing a $t$-fold convolution to characterize the distribution of $W_t$, or equivalently that of $S_t=K_1+\dots+K_t$ as in the proof of \cref{thm:FUExpMultinomial}. Namely, writing 
\begin{equation}\label{eq_compFUP2}
        \mathcal{F}_U(t) =  \sum_{\bs{s} \in \bO_t} \P(S_t = \bs{s})^2
\end{equation}
where $\bO_t \subset \Z^N$ is the set of all possible values of $S_t$. \cref{fig:FEexample} shows the  frame potential and expressiveness for the circuit shown in \cref{fig:ccircuit} ($n=4$, $N=5$) compared with their numerical approximations $\tilde{\mathcal{F}}_U$ and $\tilde{\mathcal{E}}_U$ of \cref{prop:FEapproxComm}. As stated in \cref{thm_CLTBound} and 
\cref{cor_exp}, the approximations improve as $t$ increases, which for this small circuit occurs at very small values of $t$.

\begin{figure}
    \centering
   \subfloat[]{\includegraphics[scale=0.6]{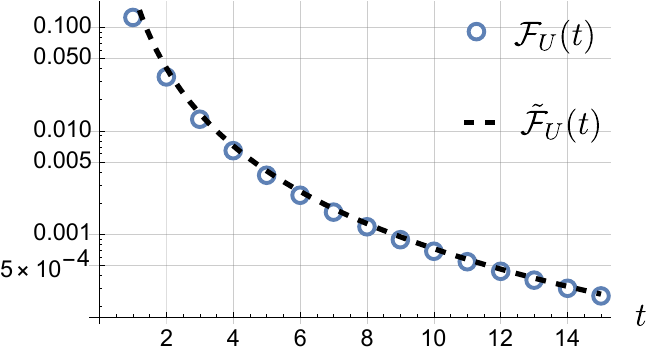}} \newline
   \subfloat[]{\includegraphics[scale=0.6]{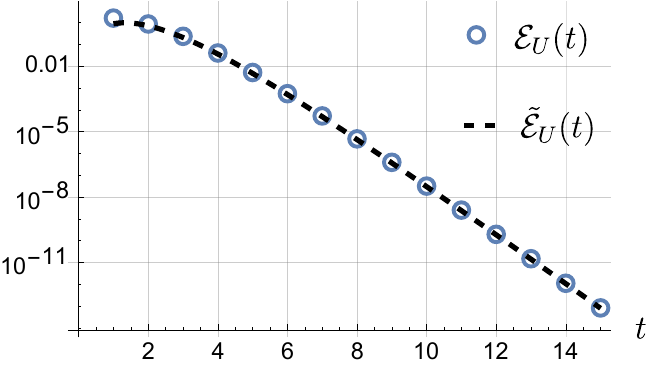}}
   \caption{Comparison in logarithmic scale between (a) the exact frame potential $\mathcal{F}_U(t)$ and (b) expressiveness $\mathcal{E}_U(t)$ with their corresponding approximations $\tilde{\mathcal{F}}_U$ and $\tilde{\mathcal{E}}_U$ for the circuit in \cref{fig:ccircuit} and $1\leq t \leq 15$. The exact quantities were computed using \cref{eq_compFUP2}.}
   \label{fig:FEexample}
\end{figure}

Another option for computing $\mathcal{F}_U$ in this framework is through sampling methods.  Note however that \cref{thm_main} is not directly useful in this regard because the random walk $W_t$ is transient and the event $[W_t = 0]$ is extremely rare. A common solution for these cases is to proceed with sequential importance sampling, namely constructing a random walk $\tilde{W}_t$ based on the increments of $W_t$, but much more likely to reach the origin. See~\citet{Tokdar2010}. Refer to the notation of \cref{def_WD} and consider for example $\tilde{W}_t = \tilde{D}_1 + \dots + \tilde{D}_t$ constructed using Algorithm 1. There, given~$\tilde{W}_s$, the increment $\tilde{D}_{s+1}$ is chosen by assigning a higher probability to those vectors $\bd$ with lowest inner product with~$\tilde{W}_s$. The resulting random walk tends to return to the origin, and once there, gets absorbed. The frame potential can thus be estimated from $M$ realizations $\{\tilde{W}^{(m)}\}_{m=1}^M$ of $\tilde{W}$ as
\begin{equation}\label{eq_FUimpSamp}
    \mathcal{F}_U(t) \approx \frac{1}{M} \sum_{m=1}^M w(\tilde W^{(m)}) \IND{[\tilde{W}^{(m)}_t = 0]}
\end{equation}
where $w(\tilde W^{(m)})$ denotes the importance weight of the $m$-th path of $\tilde{W}$. \cref{fig:FUimpSamp} shows the results of estimating $\mathcal{F}_U$ for the example circuit in \cref{fig:ccircuit} with Algorithm 1 and \cref{eq_FUimpSamp} using different values of $M$. 

\begin{figure}
    \centering
    \includegraphics[scale=0.6]{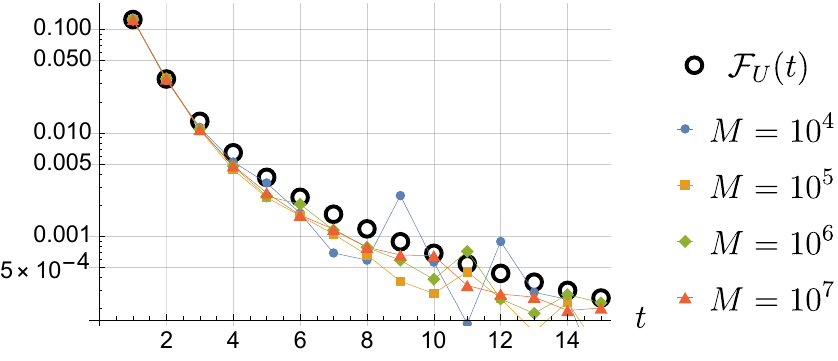}
    \caption{Results of estimating the frame potential of the circuit in \cref{fig:ccircuit} ($n=4,N=5$) using the importance sampling Algorithm 1 with different values of $M$. The results are compared in logarithmic scale with the exact calculation of $\mathcal{F}_U(t)$ in \cref{eq_compFUP2}}
    \label{fig:FUimpSamp}
\end{figure}

\begin{algorithm}
\caption{\\ Construction of $\tilde{W}_t$ via Sequential Importance Sampling}
    \begin{algorithmic}[1]
        \REQUIRE $\bs{p} = \{\P(D=\bd_l), l =1,\dots,L\}$
        \STATE sample $D_1 \sim \bs{p}$
        \STATE $\tilde{W}_1 \gets D_1$
        \FOR {$s \gets 1$ to $t-1$}
            \IF {$\tilde{W}_s = 0$}
                \STATE $\tilde{W}_{s+1} \gets 0$
            \ELSE
                \STATE $\{\tilde{\bd}_1,\dots,\tilde{\bd}_L\} \gets$ sort $\{\bd_1,\dots,\bd_L\}$ by $\bd_l \cdot \tilde{W}_t$ \STATE in ascending order
                \STATE $\tilde{\bs{p}} \gets$ sort $\bs{p}$ in descending order
                \STATE sample $\tilde{D}_{s+1} \gets \tilde{\bd}_l$ with probability $\tilde{\bs{p}}_l$
                \STATE $\tilde{W}_{s+1} \gets \tilde{W}_s + \tilde{D}_{s+1}$
            \ENDIF
        \ENDFOR
    \end{algorithmic} 
\end{algorithm}

For circuits with all rotations, $N=2^n-1$, we can estimate the frame potential by Monte Carlo sampling of the expected value in \cref{thm:FUExpMultinomial} over  realizations of the process $B$. \cref{fig:FallRots} compares the results of this method with the theoretical approximate  $\tilde{\mathcal{F}}_U(t)$. Note that by \cref{thm:maxminvol}, the volume can be computed as $V_U = (2^n)^{2^{n-1}}$. This approach provides a good estimate for the frame potential because the standard deviation of the multinomial function in \cref{eqn:FUExpMultinomial} decreases with $t$ and is of the same order of magnitude as $\mathcal{F}_U(t)$. 

\Figure[t!](topskip=0pt, botskip=0pt, midskip=0pt)[scale=0.6]{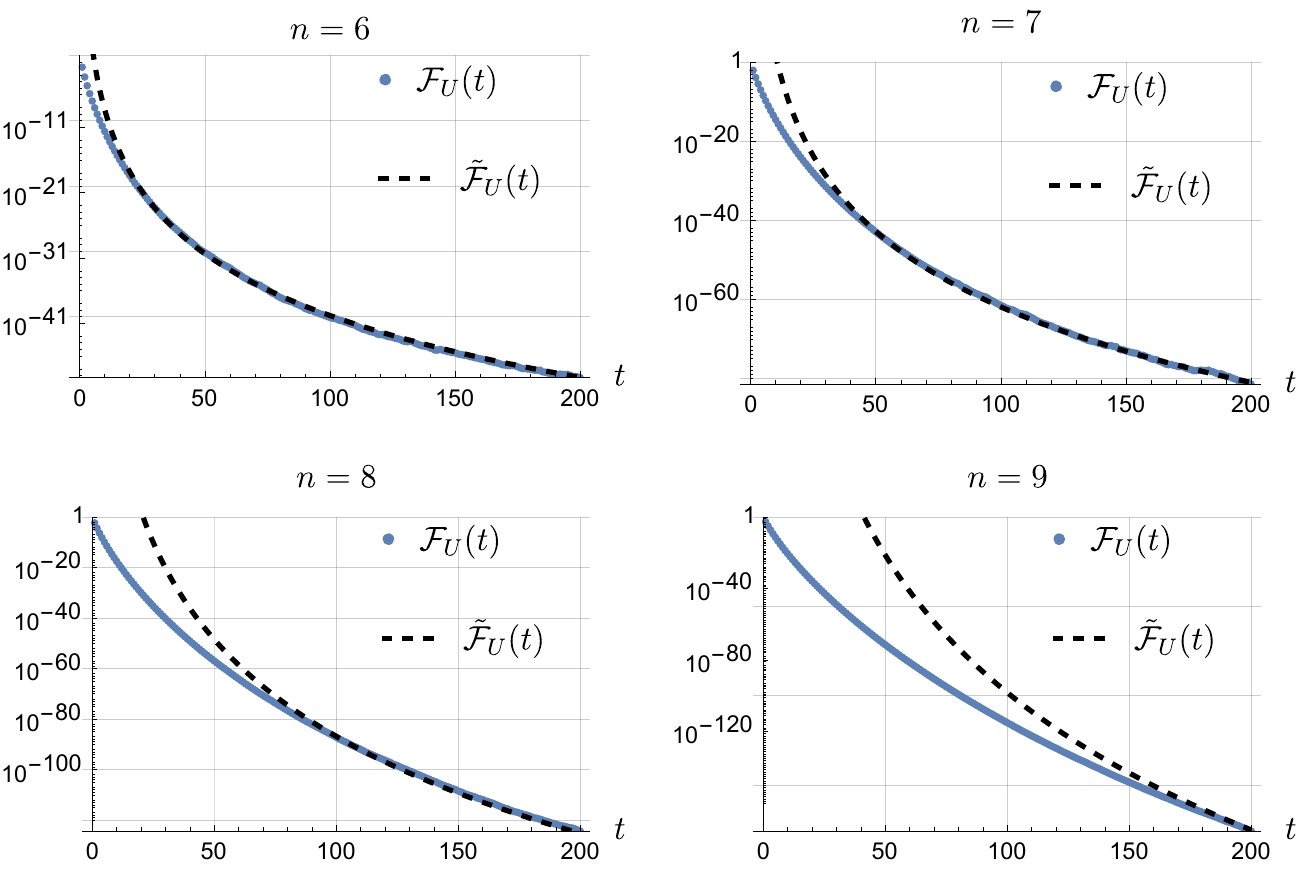}
{Comparison in logarithmic scale between $\mathcal{F}_U(t)$ and the value of $\tilde{\mathcal{F}}_U(t)$ in \cref{prop:FEapproxComm} for the circuit with all rotations ($N=2^n-1$) and four different values of the number $n$ of qubits. The values of the frame potential were estimated by sampling the expectation in \cref{thm:FUExpMultinomial} with 100 realizations of the process $B_t$ with $1\leq t \leq 200$.}  \label{fig:FallRots}

\cref{thm_CLTBound} indicates that as $N$ increases, for fixed $t$, the approximation between of $\tilde{\mathcal{F}}_U(t)$ to $\mathcal{F}_U(t)$ improves. Note that the opposite effect is observed in \cref{fig:FallRots} as the circuit increases in size. This occurs because when compared using a logarithmic scale, one gets
\begin{equation*}
   \left|  \log(\mathcal{F}_U(t)) - \log(\tilde{\mathcal{F}}_U(t)) \right| \approx \left |\frac{\mathcal{F}_U(t)}{\tilde{\mathcal{F}}_U(t)} -1 \right| < \frac{(4 \pi)^{N/2} c}{V_U ~ t}
\end{equation*}
by using a first order Taylor expansion and \cref{thm_CLTBound}. The decreasing behavior on the approximation error \[|\mathcal{F}_U(t) - \tilde{\mathcal{F}}_U(t)|\] is shown in \cref{fig:allRotsError}. The plots show evidence of the dependence of the constant $c$ in \cref{thm_CLTBound} on the size of the circuit. 

\begin{figure}
    \centering
    \includegraphics[scale=0.6]{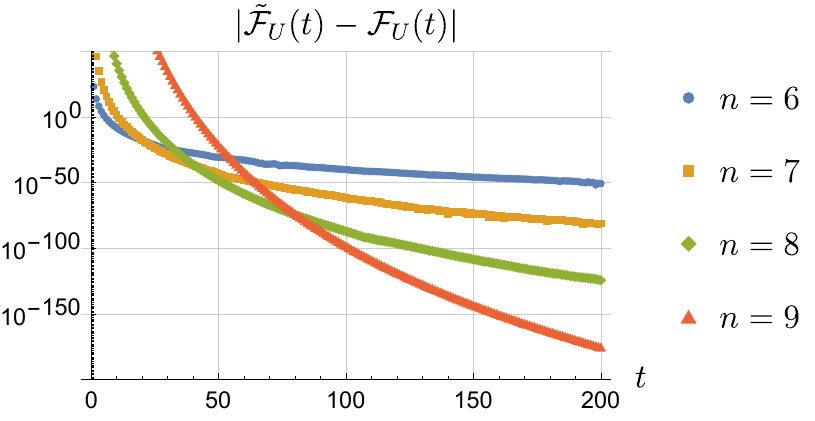}
    \caption{Logarithmic plot of the error in the approximation $|\mathcal{F}_U(t) - \tilde{\mathcal{F}}_U(t)|$ for the four circuits with all rotations ($N=2^n-1$) depicted in \cref{fig:FallRots}. The curves fit precisely the theoretical bound $c t^{-N/2+1}$ in \cref{thm_CLTBound} with $c = 56.1, 153.5, 390.1$ and $953.1$ respectively. }
    \label{fig:allRotsError}
\end{figure}

We now undertake a more comprehensive analysis of the behavior of the frame potential through the approximate formula $\tilde{\mathcal{F}}_U(t)$ obtained in \cref{thm_CLTBound}, and in particular, of the lattice volume $V_U$. First note that by \cref{lem_minor} and \cref{thm:maxminvol}, if we denote by $v_U$ the volume of the lattice $\lat(\frac{1}{2} W)$, then
\begin{equation}\label{def_vU}
    V_U = 2^N v_U, \quad v_U=|\det((\bU \bA)_{\modtwo}^*)|
\end{equation}
where $1\leq v_U \leq 2^{-N}(N+1)^{\frac{N+1}{2}}$.

For small values of $n$ we can actually enumerate all possible circuits and characterize the possible values of $v_U$. \cref{fig:allVols}a shows the results of that exercise for all 32192 circuits with $n=4$ qubits and $N\in \{4,\dots,15\}$ rotations. For $n=6$ qubits, there are close to $10^{19}$ possible circuits wih $N\in\{6,\dots,63\}$. In order to explore the values of $v_U$ for this case, we resort to randomly sampling circuits of each $N$ and tallying them according to their volume. The results are shown in \cref{fig:allVols}. In practice, $V_U$ or $v_U$ can be computed by first applying a lattice reduction algorithm on the rows of $(\bO -\bs{1})$ or $(\bU \bA)_{\modtwo}$, and then computing the determinant.

\begin{figure*}
    \centering
   \subfloat[]{\includegraphics[scale=0.6]{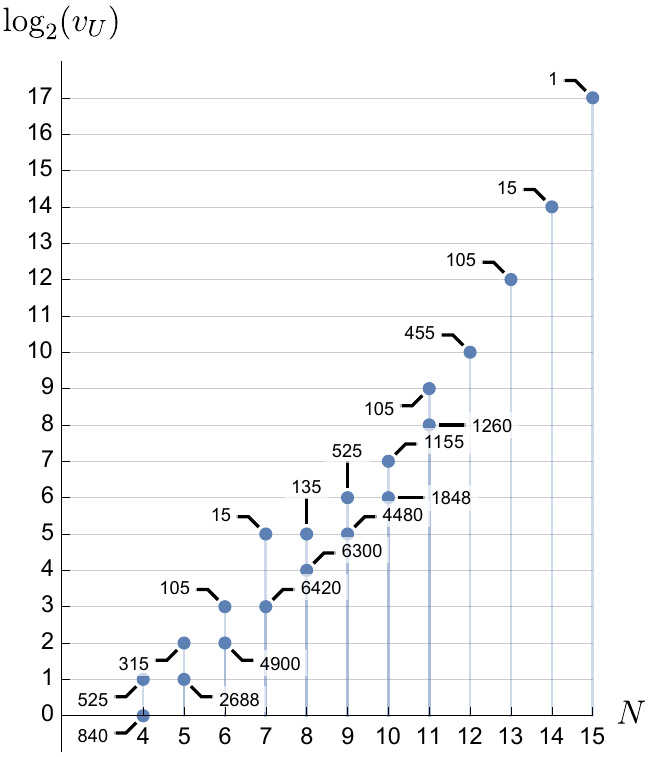}}
   \subfloat[]{\includegraphics[scale=0.6]{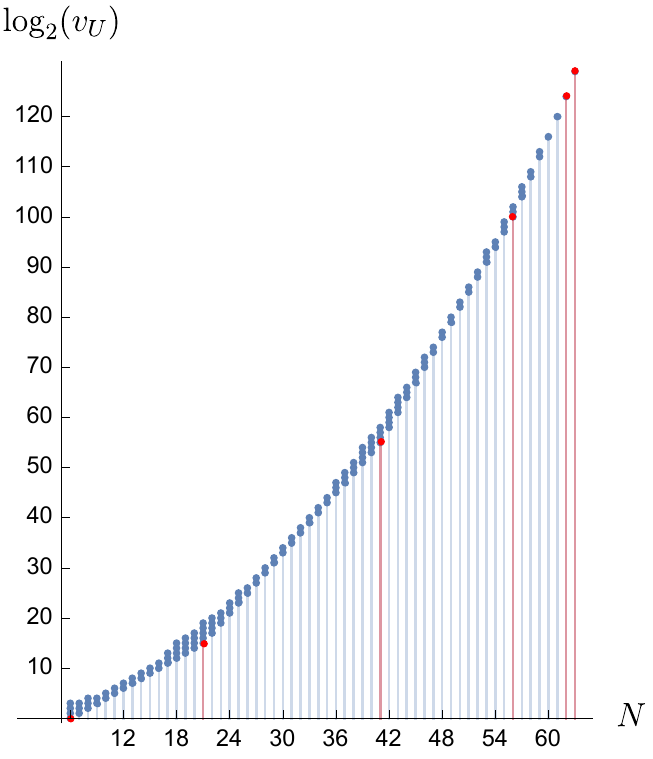}}
   \caption{Values of $v_U$ for circuits having all possible values of $N$ for (a) $n=4$ and (b) $n=6$. In all cases, $v_U$ is a power of two. In (a) all possible circuits are listed and the callout numbers indicate how many circuits of a given $N$ have each $v_U$ (for example, with $N=5$ rotations there are a total of 3003 possible circuits with 2688 and 315 circuits having $v_U = 1$ and $v_U = 4$, respectively). In (b) we sampled 1000 random circuits for each $N$ and tallied the observed values of $v_U$. The red vertical points mark circuits that use all rotations of up to $n'$ qubits, with $n'=1,\dots,n$.}
   \label{fig:allVols}
\end{figure*}

As a result of the wide variability of the lattice volume, the frame potential varies over several orders of magnitude when considering circuits of different $N$ for a given number $n$ of qubits. \cref{fig:FUtildeBounds} shows the approximate frame potential $\tilde{\mathcal{F}}_U$ for $n=6$ and different values of $N$, along with bounds derived from \cref{thm:maxminvol}. The range of variation of $V_U$ for a given $N$ is very small when compared to the range of possible values of the volume across multiple $N$s. Further, the tightness of the bounds for $\tilde{\mathcal{F}}_U$ vary greatly with $N$.  

\Figure[t!](topskip=0pt, botskip=0pt, midskip=0pt)[scale=0.6]{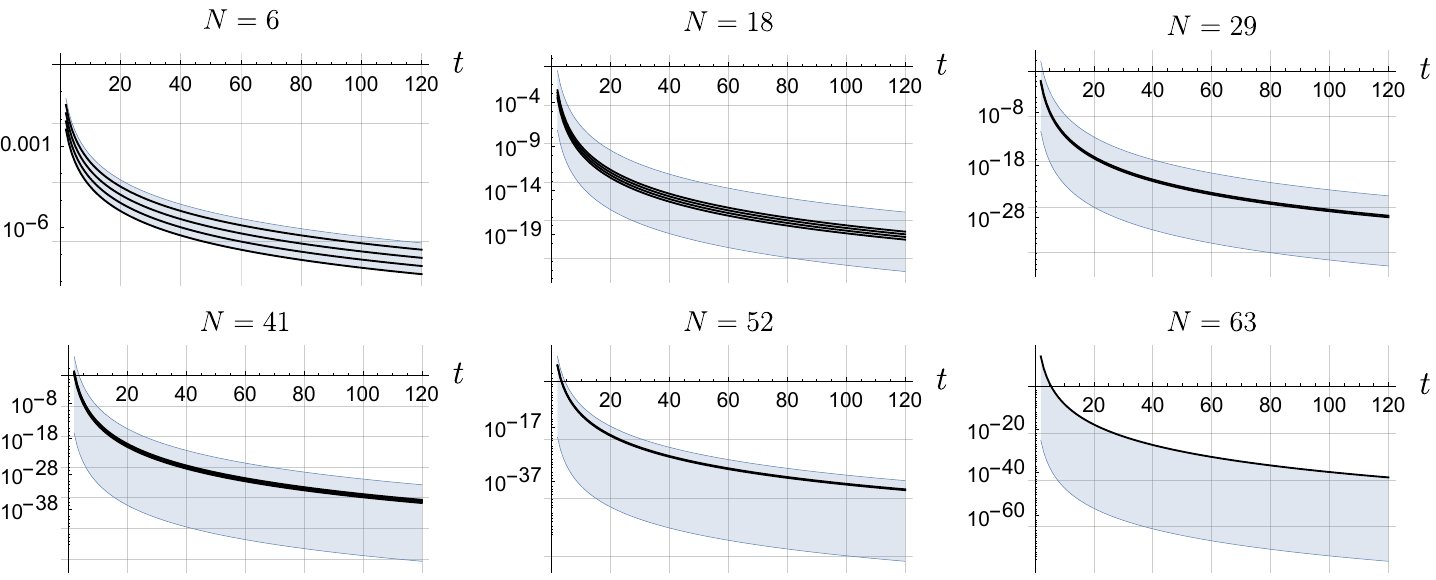}
{Logarithmic plot of $\tilde{\mathcal{F}}_U(t)$ for circuits of selected $N$s from \cref{fig:allVols}b. Each plot contains as many curves for each $N$ as values where observed for the volume. The blue area marks the bounds on $\tilde{\mathcal{F}}_U(t)$ corresponding to the bounds for $V_U$ in \cref{thm:maxminvol}. Note from \cref{fig:allVols}, case $n=6$, that  $\tilde{\mathcal{F}}_U(t)$ gives accurate information about the behavior of $\mathcal{F}_U(t)$ for all values of~$t$.}
    \label{fig:FUtildeBounds}

As a practical matter, one would like to construct the most expressive quantum circuit for a given amount of resources. Given $n$, \cref{cor:monotone} indicates that having more operators always yields more expressive circuits. It is thus important to know by how much $\mathcal{E}_U$ increases as $N$ increases, which by \cref{thm_CLTBound} can be related to the behavior of the smallest value of $v_U$ as a function of $N$. Consider operators $U_N$ and $U_{N+1}$ acting on $n$ qubits with $N$ and $N+1$ rotations, respectively. We can approximate, for large enough $t$,
\begin{align}\label{eq_EUdimRet}
    1 \leq \frac{\mathcal{E}_{U_{N+1}}(t)}{\mathcal{E}_{U_{N}}(t)} \approx  \sqrt{\pi t}~ \frac{v_{U_{N}}}{v_{U_{N+1}}}.
\end{align}
\cref{fig:rationsEut}a shows the behavior of the last factor in \cref{eq_EUdimRet} computed over circuits of minimum value of $v_U$ of each $N$ in the case $n=6$ from \cref{fig:allVols}b. According to this behavior, there is a `diminishing return' effect on the expressiveness as more rotations are added, and for all $t$. This effect can also be observed in the approximate expressiveness plot in \cref{fig:rationsEut}b. The values where the gain in expressiveness is reduced the most occur precisely at circuits that include ``all rotations of up to $n'$ qubits'' with $n' =1,\dots,n$, a commonly used choice. The red lines in \cref{fig:allVols}b show these circuits, indicating that their value of $v_U$ is minimum for their corresponding $N$. In fact, we observe that for any $N$, one way to ensure minimum volume is to include all rotations up to $n'$ qubits, for the highest possible $n'$. For example, and referring to \cref{fig:allVols}b, every circuit with $N\in \{21,\dots,40\}$ that includes all 21 rotations of up to $n'=2$ qubits, has minimum volume for its corresponding $N$.

\begin{figure*}
\centering
   \subfloat[]{\includegraphics[scale=0.6]{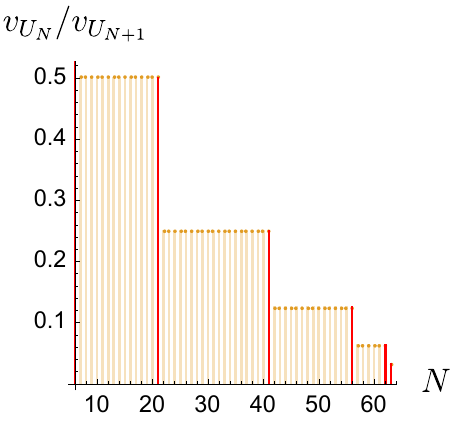}}
   \subfloat[]{\includegraphics[scale=0.6]{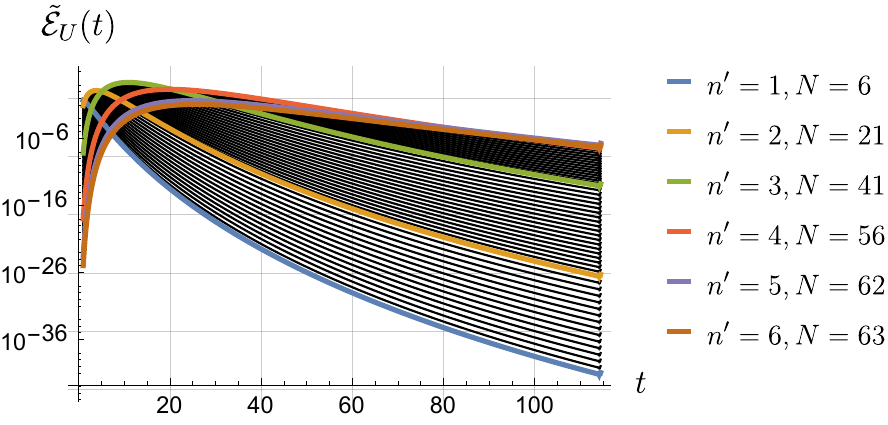}}
   \caption{(a) Ratios between the minimum values of $v_U$ for circuits of consecutive values of $N$, $N=1,\dots,2^n-1$, for $n=6$ as in \cref{fig:allVols}b. The red lines mark the values of $N$ of circuits with all rotations up to $n'$ qubits, $n'=1,\dots,n$. (b) Logarithmic plot of the approximate expressiveness for circuits of minimum value of $v_U$ of each value of $N$ with $n=6$. Circuits with all rotations up to $n'$ qubits are marked in color.}
   \label{fig:rationsEut}
\end{figure*}

\section{Discussion and Outlook }
\label{sec:outlook}

In this study, we conducted a detailed analysis of the frame potential and expressiveness in parametric commutative circuits with integer eigenvalues. Our techniques are probabilistic, based on the observation that for the corresponding operator $U(\btheta)$, the Fourier expansion of the function $f_U(\btheta) = \langle \psi_0 | U(\btheta) | \psi_0 \rangle$ in \cref{def_fU} can be written as the characteristic function of a random vector $K$ with integer entries. As a consequence, \cref{thm_main} shows that the frame potential can be expressed as $\mathcal{F}_U(t) = \P(W_t = 0)$ where $W$ is a random walk on $\Z^N$, with $N$ being the number of Hamiltonian operators in $U$. Our key result, \cref{thm_CLTBound}, provides an approximate formula for the frame potential that is useful for estimating the expressiveness. It also reveals the existence of a quantity $V_U$ that greatly influences the order of magnitude of expressiveness. In our formulation, $V_U$ is the volume of the sub-lattice of $\Z^N$ where the random walk $W$ takes values.

Focusing on commutative circuits composed of Pauli-$Z$ rotations, \cref{sec_CommC} uses the probabilistic representation to understand the relationship between the architecture of the circuit, its frame potential, and expressiveness. \cref{thm_distK} and \cref{prop_relate2A} explicitly characterize the distribution of $K$ and the lattice of $W$ in terms of the matrix $\bA$ which explicitly encodes the circuit's architecture. As a result, in \cref{lem_minor} and \cref{thm:maxminvol} we obtain expressions to compute and bound $V_U$. In addition, for circuits that include all rotations with $n$, we derive a further probabilistic representation for $\mathcal{F}_U$ that generalizes the combinatorial approach of~\cite{Bennink2023}, and provides a useful sampling method for the calculation of expressiveness.

The numerical examples presented in \cref{sec:calculations} illustrate that with Pauli-$Z$ rotations, some of the practical insights that can be obtained using our results. We provide different approaches to compute or estimate the frame potential to different levels of accuracy, and we show how our approximate formula $\tilde{\mathcal{E}}_U(t)$ allows for accurate calculation of the expressiveness for $t$-values that were previously unfeasible.

The volume $V_U$, or its associated $v_U$ in \cref{def_vU}, appear to be natural indicators of expressiveness for commutative quantum circuits. Its computation is generally straightforward given the spectrum $\bO$ and, as shown in \cref{sec:calculations}, its value can vary over several orders of magnitude. The results shown in \cref{fig:allVols}, along with others not included here, allow us to make some observations and conjectures. First, we observe that $v_U$ seems to always be a power of two, growing to very large numbers as $N$ increases. However, $v_U$ does not take all possible powers of two, and for a given $N$, multiple values of $v_U$ can be possible, spanning various orders of magnitude. This is significant because it means that given $n$ and $N$, the choice of rotations can directly affect the expressiveness, and one would like criteria to always select an architecture with minimum volume. Along these lines, we observe that for any given number of resources $n,N$, it is sufficient to design the circuit to have all rotations up to some number of qubits, which in practice is the easiest architecture to implement.  

A rigorous characterization of the possible values of $V_U$  is, by \cref{lem_minor}, equivalent to characterizing the minors of the Hadamard matrix $\bs{H}_{2^n}$ that arise from arbitrarily selecting columns, and rows generating a given lattice. Calculating minors from general Hadamard matrices is a problem of great interest dating back to 1907. In fact, the possible values of minors of order $n-4$ of a Hadamard matrix can be computed~\citep[see][ Theorem 5]{Kravv2016}, and theoretically generalized to $n-j$ minors, for any $j<n$. However, this generalization is currently impossible to implement due to the complexity of the computation. For $N\in \{2^n-4,\dots,2^n-1\}$, existing characterizations of the minors of $\bs{H}_{2^n}$ yield the specific possible values that $V_U$ can take. Further explorations of the behavior of the volume $V_U$ will reveal interesting and informative patterns on how expressiveness can be managed through careful selection of circuit architecture.

The results in \cref{sec_General} apply to any commutative quantum circuit of the form \eqref{def_U}. However, their applications to obtain estimates to the frame potential via \cref{thm_CLTBound} require the characterization of the distribution of $K$, which in \cref{sec_CommC} is carried out for circuits comprising only of Pauli rotations. The analysis for the case when $\{H_j\}_{j=1}^N$ are commuting Hamiltonians in the Pauli group $\mathbb{S}_n$ is the subject of a forthcoming note. \citep[See][]{Yu2024}.    

A natural question that arises is whether the probabilistic framework presented here can be extended to the case of parametric quantum circuits built with Hamiltonian operators that are not simultaneously diagonalizable, and hence non-commutative. This extension would allow, as was done here, for the use of techniques in probability theory to analyze a wider class of parametric quantum circuits. This is the subject of future research but, to close, we now present some interesting preliminary observations regarding this matter.

\subsection{Non-commutative Circuits}

Consider a unitary operator $U(\btheta)$ as in \cref{def_U} where $\{H_j\}_{j=1}^N$ are not commuting Hamiltonians. In this case, the simplification of the fidelity in \cref{eq_FU2Dto1D} is not possible and we must instead analyze the expression
\begin{equation}\label{eq_fUNonComm}
    f_U(\btheta, \btheta') = \langle \psi_0 | U(\btheta)^\dagger U(\btheta') | \psi_0 \rangle, \quad \btheta,\btheta' \in [-\pi,\pi]^N.
\end{equation}
A $2N$-dimensional Fourier series for $f_U$ can then be computed as
\begin{equation}\label{eq: Fourier_sum_noncomm}
    f_U(\btheta, \btheta') =\sum_{\bk \in \bO, \bk'\in\bO'} \hat{f}_U(\bk, \bk') e^{i \btheta \cdot \bk} e^{ i\btheta'\cdot \bk'},
\end{equation}
where the coefficients are
\begin{align*}
    \hat{f}_U&(\bk,\bk')\\ 
    &= \frac{1}{(2 \pi)^{2N}} \iint\limits_{[-\pi,\pi]^{2N}}  f_U(\btheta,\btheta') e^{-i \btheta'\cdot \bk'} e^{-i \btheta\cdot \bk} \ud \btheta \ud \btheta'.
\end{align*}
A probabilistic representation analogous to the one derived in \cref{sec_General} will follow whenever $\hat{f}_U(\bk,\bk') \in [0,1]$ for all pairs $(\bk,\bk') \in \bO\times \bO'$ and $\sum_{\bk \in \bO, \bk'\in\bO'} \hat{f}_U(\bk, \bk') = 1$. In this case, we can write
\begin{equation}\label{eq:ProbRepNC}
     f_U(\btheta, \btheta') = \EXP e^{i \btheta \cdot K} e^{i \btheta' \cdot K'}
\end{equation}
where $(K,K')$ is a random $2N$-dimensional vector on $\bO\times \bO'$ with distribution $\P(K = \bk, K' = \bk') = \hat{f}_U(\bk, \bk')$. 

To establish the possibility of the representation \cref{eq:ProbRepNC} for non-commutative circuits, we considered all circuits of the form shown in \cref{fig:example_noncommGen}. Namely, with $n=N=2$ and Hamiltonians given by \cref{def_HPauli} where now each $O^{(j)}_m$ can be the identity operator or one of the three Pauli matrices. There are a total of 256 possible such circuits, out of which 120 are non-commutative in the sense that $H_1H_2 \neq H_2 H_1$.

\Figure[t!](topskip=0pt, botskip=0pt, midskip=0pt){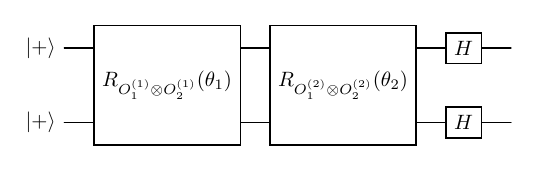}
{Circuits with $n=2$ qubits and $N=2$ rotations, and $\psi_0 = |+\rangle$. The Hamiltonians are $H_j =O^{(j)}_1\otimes O^{(j)}_2$ with $O^{(j)}_m \in \{I_2,X,Y,Z\}$ for each $j,m=1,2$.\label{fig:example_noncommGen}}

By explicitly computing the Fourier series representation for each non-commutative circuit, we identified that 72 of them have the representation \cref{eq:ProbRepNC}, whereas the remaining 48 did not because at least one of the coefficients $\hat{f}_U(\bk, \bk')$ is not a real number in $[0,1]$. \cref{fig:example_noncomm} shows one example from each class. 

\begin{figure}
    \centering
   \subfloat[]{\includegraphics[scale=1]{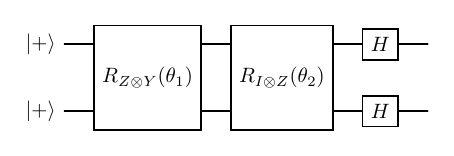}} \newline
   \subfloat[]{\includegraphics[scale=1]{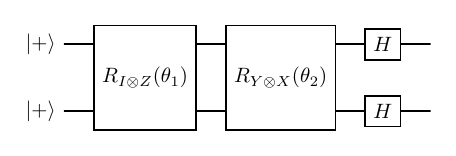}}
   \caption{Non-commutative circuit examples with $n=N=2$ that (a) admit and (b) do not admit, the probabilistic representation in \cref{eq:ProbRepNC} for their quantum expectation.}
   \label{fig:example_noncomm}
\end{figure}

For the circuit on the left in \cref{fig:example_noncomm},
\begin{align*}
f_{U_{\text{left}}}(\btheta,\btheta') &= \frac{1}{4} e^{-i \left(-\theta '_1-\theta '_2+\theta _1+\theta
   _2\right)}+\frac{1}{4} e^{i \left(-\theta '_1-\theta '_2+\theta
   _1+\theta _2\right)}\\&+\frac{1}{4} e^{-i \left(-\theta '_1+\theta
   '_2+\theta _1-\theta _2\right)}+\frac{1}{4} e^{i \left(-\theta
   '_1+\theta '_2+\theta _1-\theta _2\right)},
\end{align*}
which has $\bO =\bO'=\{ (-1,-1), (1,1), (-1,1), (1,-1) \}$. The circuit on the right has
\begin{align*}
f_{U_{\text{right}}}(\btheta,\btheta') & = 
\frac{1}{4} e^{-i \left(-\theta '_1-\theta '_2+\theta _1+\theta
   _2\right)}
   +\frac{1}{4} e^{i \left(-\theta '_1-\theta '_2+\theta
   _1+\theta _2\right)}
   \\&-\frac{i}{4} e^{-i \left(\theta '_1-\theta
   '_2+\theta _1+\theta _2\right)}
   -\frac{i}{4} e^{i \left(\theta
   '_1-\theta '_2+\theta _1+\theta _2\right)}\\
   & +\frac{1}{4} e^{-i
   \left(-\theta '_1+\theta '_2+\theta _1-\theta _2\right)}
   +\frac{1}{4}
   e^{i \left(-\theta '_1+\theta '_2+\theta _1-\theta
   _2\right)}
   \\&+\frac{i}{4} e^{-i \left(\theta '_1+\theta '_2+\theta
   _1-\theta _2\right)}
   +\frac{i}{4} e^{i \left(\theta '_1+\theta
   '_2+\theta _1-\theta _2\right)},
\end{align*}
which clearly does not admit the representation \cref{eq:ProbRepNC}. These examples highlight the need for additional investigation to determine which (sub)-classes of non-commutative circuits naturally lend itself to a probabilistic interpretation, as in the commutative case. In particular, we seek to understand under what circumstances or algebraic structures does a non-commuting circuit of Pauli rotations results in a Fourier series with a probability mass function as its weights.

\section*{Acknowledgements}
Support for this work came from the DOE Advanced Scientific Computing Research (ASCR) Accelerated Research in Quantum Computing (ARQC) Program under field work proposal 3ERKJ354.

\bibliographystyle{IEEEtranN}

\bibliography{expressiveness.bib}

\EOD

\end{document}